\documentclass[journal,a4paper]{IEEEtran}
\pdfoutput=1
\IEEEoverridecommandlockouts

\usepackage{graphicx}

\usepackage{cite}
\usepackage{amsmath}
\usepackage[caption=false]{subfig}

\usepackage{amssymb}
\usepackage{amsfonts}

\usepackage{algorithmic}
\usepackage{algorithm}

\newtheorem{proposition}{\bf Proposition}

\newtheorem{definition}{\bf Definition}
\newtheorem{remark}{Remark}

\newcommand{\sched}{\lambda}

\newcommand{\pow}{p}

\newcommand{\rate}{r}
\newcommand{\arrival}{a}
\newcommand{\channel}{h}

\newcommand{\tNx}{(t,\vectxx)}

\newcommand{\auxSched}{\upsilon}
\newcommand{\vqSched}{\Upsilon}

\newcommand{\game}{\set{G}}
\newcommand{\util}{\Gamma}
\newcommand{\policy}{p}
\newcommand{\mass}{\rho}

\newcommand{\noise}{\sigma^2}

\newcommand{\isd}{\text{ISD}}

\newcommand*{\myfigfactorx}{0.9}

\DeclareMathOperator*{\argmax}{arg\,max}
\DeclareMathOperator*{\argmin}{arg\,min}

\newcommand{\expect}{\mathbb{E}\,}

\newcommand{\vect}{\boldsymbol}

\newcommand{\vectSubTwoTime}[3]{ {#1}_{1}(#3),\dots,{#1}_{#2}(#3) }

\newcommand{\seta}[1]{1,\dots,#1}
\newcommand{\set}[1]{\mathcal{#1}}
\newcommand{\size}[1]{|\set{#1}|}

\newcommand{\vectx}{\vect{x}}
\newcommand{\vectxx}{\breve{\vect{x}}}
\newcommand{\vecty}{\vect{y}}
\newcommand{\vectY}{\vect{Y}}

\newcommand{\one}{\mathbf{1}}
\newcommand{\zero}{\mathbf{0}}

\newcommand{\texth}[1]{\!^\text{#1}}

\newcommand{\tran}{^\dag}

\usepackage{forloop}
\newcounter{loopcntr}
\newcommand{\rpt}[2][1]{%
	\forloop{loopcntr}{0}{\value{loopcntr}<#1}{#2}%
}

\newcommand{\subgroup}[1]%
{\rlap{\smash{%
	\newcount\cnt%
	\cnt \numexpr#1\relax%
	\advance\cnt -1\relax%
	$\tabcolsep=.1em\begin{tabular}[t]{|l}\multicolumn{1}{l}{}\\%
	\rpt[\cnt]{\\}
	\\\hline\end{tabular}$%
}}}

\newcounter{myRefCount}

\begin{document}

\newlength\figurewidth
	
\title{Ultra Dense Small Cell Networks: Turning Density into Energy Efficiency
}

\author{
	Sumudu Samarakoon,~\IEEEmembership{Student~Member,~IEEE,}
	Mehdi Bennis,~\IEEEmembership{Senior~Member,~IEEE,}
	Walid Saad,~\IEEEmembership{Senior~Member,~IEEE,}
	M\'{e}rouane Debbah,~\IEEEmembership{Fellow,~IEEE}
	and~Matti~Latva-aho,~\IEEEmembership{Senior~Member,~IEEE}
	\thanks{
		This work is supported by the TEKES grant 2364/31/2014, the Academy of Finland (289611), the 5Gto10G project, the SHARING project under the Finland grant 128010, the U.S. National Science Foundation (NSF) under Grants CNS-1460333, CNS-1460316, and CNS-1513697, and the ERC Starting Grant 305123 MORE (Advanced Mathematical Tools for
		Complex Network Engineering).
		}%
	\thanks{
		S. Samarakoon, M. Bennis and~M.~Latva-aho are with the Department of Communications Engineering, University of Oulu, Finland (e-mail: \{sumudu,bennis,matti.latva-aho\}@ee.oulu.fi).
		
		W. Saad is with Wireless@VT, Bradley Department of Electrical and Computer Engineering, Virginia Tech, Blacksburg, VA (email: walids@vt.edu).
		
		M. Debbah is with Mathematical and Algorithmic Sciences Lab, Huawei France R\&D, Paris, France, (email: merouane.debbah@huawei.fr) and  with the Large Systems and Networks Group (LANEAS), CentraleSup\'{e}lec, Universit\'{e} Paris-Saclay, 3 rue Joliot-Curie, 91192 Gif-sur-Yvette, France.
		}
	}

\maketitle
\nopagebreak[4]
\vspace{-20pt}
\begin{abstract}

In this paper, a novel approach for joint power control and user scheduling is proposed for optimizing energy efficiency (EE), in terms of bits per unit energy, in ultra dense small cell networks (UDNs). 
Due to severe coupling in interference, this problem is formulated as a dynamic stochastic game (DSG) between small cell base stations (SBSs). 
This game enables to capture the dynamics of both the queues and channel states of the system. 
To solve this game, assuming a large homogeneous UDN deployment, the problem is cast as a \emph{mean-field} game (MFG) in which the MFG equilibrium is analyzed with the aid of low-complexity tractable partial differential equations.
Exploiting the stochastic nature of the problem, user scheduling is formulated as a stochastic optimization problem and solved using the drift plus penalty (DPP) approach in the framework of \emph{Lyapunov} optimization.
Remarkably, it is shown that by weaving notions from Lyapunov optimization and mean-field theory, the proposed solution yields an equilibrium control policy per SBS which maximizes the network utility while ensuring users' quality-of-service.
Simulation results show that the proposed approach achieves up to $70.7\%$ gains in EE and $99.5\%$ reductions in the network's outage probabilities compared to a baseline model which focuses on improving EE while attempting to satisfy the users' instantaneous quality-of-service requirements.

\end{abstract}


\section{Introduction}\label{sec:introduction}

Wireless network densification is viewed as a promising approach to enable a 1000x improvement in wireless cellular network capacity~\cite{jnl:andrews14,online:alexiou13,pap:yunas14,pap:talwar14}.
Indeed, 5G systems are expected to be ultra-dense in terms of transmitters and receivers rendering network optimization highly complex~\cite{online:alexiou13} and \cite{jnl:wu15}.
In such ultra dense network (UDN) environments, resource management problems, such as power control and user equipment (UE) scheduling, become significantly more challenging due to the scale of the network~\cite{online:alexiou13,jnl:lopez15}. 
In addition, the spatio-temporal traffic demand fluctuations in the network, the dynamics in channel conditions, and the increasing overhead due to the need for coordination will further exacerbate the challenges of resource management in UDNs.
For instance, the uncertainties in terms of queue state information (QSI) and channel state information (CSI) as well as their evolution over time will play a pivotal role in resource optimization.  

Due to the involvement of large number of devices and the severe coupling of their control parameters with one another, the current state-of-the-art approaches used in classical, small scale network deployments are inadequate to study the resource optimization in UDNs.
Thus, power control, UE scheduling, interference mitigation, and base station (BS) deployment strategies for energy-efficient UDN have been recently investigated in a number of works such as \cite{jnl:andrews14} and \cite{jnl:wu15,jnl:lbjornson15,jnl:fahad15,jnl:lopez15,pap:ternon14,pap:gotsis14,pap:li15}.
The work in \cite{jnl:lopez15} investigates the potential gains and limitations of energy efficiency (EE) for outdoor UDNs considering small cells idle mode capabilities and BS/UE densification.
In \cite{jnl:lbjornson15}, an analytical framework was proposed for designing UDNs that maximizes EE by modeling the network using stochastic geometry and obtaining new lower bound on the average spectral efficiency.
The authors in \cite{jnl:fahad15} investigate how different BS and antenna deployment strategies, for both indoor and outdoor UDNs, can impact EE.
In \cite{pap:ternon14}, a power control and UE scheduling mechanism is proposed for UDN of clustered small cells which have access to a database consisting signal-to-noise ratio values of each BS-UE link.
The work in \cite{pap:gotsis14} proposes an optimization framework which allows flexible UE association among densely deployed BSs by minimizing the signaling overhead.
The authors in \cite{pap:li15} study how to improve EE of an UDN by jointly optimizing the resource utilization and the temporal variances in the traffic load. 
These studies provide valuable insights of both performance gains and limitations of UDNs.
However, most of these works \cite{jnl:lbjornson15,jnl:fahad15,jnl:lopez15,pap:ternon14,pap:gotsis14}, ignore uncertainties in QSI, do not properly model network dynamics, and typically focus on the problems of UE scheduling and resource allocation in isolation. 
Instead, in practice, there is a need to treat those problems jointly while accounting for the QSI dynamics and uncertainty.
Moreover, none of these existing works study the behavior of the ultra dense system setting in which the number of BSs and UEs grows infinitely large.

Recently, mean-field games (MFGs) received significant attention in the context of cellular networks with large number of players~\cite{book:gueant11,book:caines14,pap:mari12,jnl:meriaux13,pap:gummadi12}.
In MFGs, a large population consists of small interacting individual players and their decision making strategies are studied~\cite{book:gueant11,book:caines14}.
Due to the size of the population, the impacts of individual states and decisions are negligible while the abstract behavior of the population is modeled by a mean-field (MF).
As a result, the MF regime allows to cast the multi-player problem into a more tractable single player problem which depends on the state distribution of the population.
In~\cite{pap:mari12}, concurrent packet transmissions among large number of transmitter-receiver pairs with the objective of minimizing transmit power under the uncertainties of QSI and CSI are investigated using MFGs.
The authors in~\cite{jnl:meriaux13} study an energy efficient power control problem in multiple access channels whereas transmitters strategies depend on their receivers, battery levels and the strategies of others.
In~\cite{pap:gummadi12}, a competition of wireless resource sharing among large number of agents is investigated as an MFG.
Despite their interesting analytical perspective, these existing works remain limited to traditional macrocellular networks and do not address the challenges of small cell networks. 
Moreover, for simplicity, these works consider models with a single UE per BS or they address multi-user scenarios with conventional UE scheduling schemes (proportional-fair, round-robin, first-in-first-out). Thus, the opportunity to improve the proposed solutions for multi-user cases by smart UE scheduling has not been treated yet.

When considering multiple UEs associated with SBS, UE scheduling directly affects the performance of each SBS.
Due to the dynamic nature of channels and arbitrary arrivals, the UE scheduling process becomes a stochastic optimization problem~\cite{book:neely10,jnl:neely10,jnl:huang11,book:georgiadis06}.
To solve such stochastic optimization problems while ensuring queue stability, the \emph{dual-plus-penalty} (DPP) theorem from the Lyapunov optimization framework can be used~\cite{book:neely10,jnl:neely10,jnl:huang11,pap:neely10,book:georgiadis06,jnl:bethanabhotla13}.
Lyapunov DPP approach simplifies the stochastic optimization problem into multiple subproblems which can be solved at each time instance.
Remarkably, this simplification ensures the convergence to the optimal solution of the original stochastic optimization problem~\cite{book:neely10,jnl:neely10,jnl:huang11,book:georgiadis06}.
The work in \cite{jnl:huang11} presents a utility optimal scheduling algorithm for  processing networks based on Lyapunov DPP approach.
In~\cite{pap:neely10}, Lyapunov DPP method is used to solve stochastic optimization problem with non-convex objective functions.
In~\cite{jnl:bethanabhotla13}, a Lyapunov DPP framework is used to dynamically select service nodes for multiple video streaming UEs in a wireless network.
It can be noted that the Lyapunov DPP approach is a powerful tool for solving stochastic optimization problems.
However, to implement the above technique in decentralized control mechanisms, \emph{i)} the objective function should be chosen such that it can be decoupled among the control nodes as in \cite{jnl:huang11} or \emph{ii)} necessary bounds need to be identified for the coupling term as the authors in \cite{jnl:bethanabhotla13} evaluate a lower bound for the capacity in order to mitigate coupling due to interference.
Therefore, a suitable method to make an accurate evaluation on the coupled term (interference in wireless UDNs) needs to be investigated. 

The main contribution of this paper is to propose a novel joint power control and user scheduling mechanism for ultra-dense small cell networks with large number of small cell base stations (SBSs).
In the studied model, the objective of SBSs is to maximize their own time average energy efficiency (EE) in terms of the amount of bits transmitted per joule of energy consumption.
Due to the ultra-dense deployment, severe mutual interference is experienced by all SBSs and thus, must be properly managed.
To this end, SBSs have to compete with one another in order to maximize their EEs while ensuring UEs' quality-of-service (QoS).
This competition is cast as a dynamic stochastic game (DSG) in which players are the SBSs and their actions are their control vectors which determine the transmit power and user scheduling policy.
Due to the non-tractability of the DSG, we study the problem in the mean-field regime which enables us to capture a very dense small cell deployment.
Thus, employing a homogeneous control policy among all the SBSs in the network, the solution of the DSG is obtained by solving a set of coupled partial differential equations (PDEs) known as Hamilton-Jacobi-Bellman (HJB) and Fokker-Planck-Kolmogorov (FPK) equations~\cite{book:gueant11}.
The resulting equilibrium is known as the \emph{mean-field equilibrium} (MFE).
We analyze the sufficient conditions for the existence of an MFE in our problem and we prove that all SBSs converge to this MFE.
Exploiting the stochastic nature of the objective with respect to QSI and CSI, the UE scheduling procedure is modeled as a stochastic optimization problem.
To ensure the queue stability and thus, UE's QoS, the UE scheduling problem is solved within a DPP-based Lyapunov optimization framework~\cite{jnl:bethanabhotla13,book:neely10}.
An algorithm is proposed in which each SBS schedules its UEs as a function of CSI, QSI and the mean-field of interferers.
Remarkably, it is shown that combining the power allocation policy obtained from the MFG and the DPP-based scheduling policy enables SBSs to autonomously determine their optimal transmission parameters \emph{without coordination} with other neighboring cells.
To the best of our knowledge, \emph{this is the first work combining MFG and Lyapunov frameworks within the scope of UDNs}.
Simulation results show that the proposed approach achieves up to $70.7\%$ gains in EE and $99.5\%$ reductions in the network's outage probabilities compared to a baseline model which improves EE while satisfying users' instantaneous quality-of-service requirements.

The rest of this paper is organized as follows. 
Section~\ref{sec:system_model} presents the system model, formulates the DSG with finite number of players and discusses the sufficient conditions to obtain an equilibrium for the DSG.
In Section~\ref{sec:MFG}, using the assumptions of large number of players the DSG is cast as a MFG and solved.
The solution for UE scheduling at each SBS based on DPP framework is examined in Section~\ref{sec:formulations}.
The results are discussed in Section~\ref{sec:results} 
and finally, conclusions are drawn in Section~\ref{sec:conclusion}.

\section{System Model and Problem Definition}\label{sec:system_model}

Consider a downlink of an ultra dense wireless small cell network that consists of a set of $B$ SBSs $\set{B}$ using a common spectrum with bandwidth $\omega$.
These SBSs serve a set $\set{M}$ of $M$ UEs.
Here, $\set{M}= \set{M}_1 \cup \dots \cup \set{M}_{B}$ where $\set{M}_b$ is the set of UEs served by SBS $b\in\set{B}$.
For UE scheduling, we use a scheduling vector $\vect{\sched}_{b}(t)=\big[ \sched_{bm}(t) \big]_{\forall m\in\set{M}_b}$ for SBS $b$ where $\sched_{bm}(t)=1$ indicates that UE $m\in\set{M}_b$ is served by SBS $b$ at time $t$ and $\sched_{bm}(t)=0$ otherwise.
The channel gain between UE $m\in\set{M}_b$ and SBS $b$ at time $t$ is denoted by $\channel_{bm}(t)$ and an additive white Gaussian noise with zero mean and $\noise$ variance is assumed.
The instantaneous data rate of UE $m$ is given by:
\begin{equation}\label{eqn:datarate_ue}
	\rate_{bm}(t) = \omega\sched_{bm}(t) \log_2 \bigg( 1 + \frac{\pow_b(t)|\channel_{bm}(t)|^2}{I_{bm}(t) + \noise} \bigg),
\end{equation}
where $\pow_b(t)\in[0,\pow_b^{\text{max}}]$ is the transmission power of SBS $b$, $|\channel_{bm}(t)|^2$ is the channel gain between SBS $b$ and UE $m$, $I_{bm}(t)=\sum_{\forall b'\in\set{B}\setminus\{b\}}\pow_{b'}(t)|\channel_{b'm}(t)|^2$ is the interference term.

We assume that SBS $b$ sends $q_{bm}(t)$ bits to UE $m\in\set{M}_b$.
Thus, the time evolution of the $b$-th SBS queue, 
is given by:
\begin{equation}\label{eqn:evolution_queue}
	d\vect{q}_b(t) = \vect{\arrival}_b(t) - \vect{\rate}_{b}\big(t,\vectY(t),\vect{\channel}(t)\big) dt,
\end{equation}
where $\vect{\arrival}_b(t)$ and $\vect{\rate}_{b}(\cdot)$ are the vectors of arrivals and transmission rates at SBS $b$ and is the vector of channel gains.
Moreover, a vector of control variables $\vectY(t)=\big[\vecty_b(t),\vecty_{-b}(t)\big]$ is defined such that $\vecty_b(t)=\big[\vect{\sched}_b(t),{\pow}_b(t)\big]$ is the SBS local control vector and $\vecty_{-b}(t)$ is the control vector of interfering SBSs.
The evolution of the channels are assumed to vary according to the following known stochastic model~\cite{book:rappaport02}:
\begin{equation}\label{eqn:evolution_channels}
	d\vect{\channel}_{b}(t) = \vect{G}\big( t, \vect{\channel}_{b}(t)\big) dt + \zeta d\vect{z}_{b}(t),
\end{equation}
where the deterministic part $\vect{\channel}_{b}(t)\big)=[G\big( t, \channel_{bm}(t)\big)]_{m\in\set{M}}$ considers path loss and shadowing while the random part $\vect{z}_{b}(t)=[z_{bm}(t)]_{m\in\set{M}}$ with positive constant $\zeta$ includes fast fading and channel uncertainties.
The evolution of the entire system can be described by the QSI and the CSI as per (\ref{eqn:evolution_queue}) and (\ref{eqn:evolution_channels}), respectively.
Thus, we define $\vectx(t)= \big[ \vectx_b(t) \big]_{ b\in\set{B} } \in\set{X} $ as the state of the system at time $t$ with $\vectx_b(t)=\big[\vect{q}_b(t), \vect{\channel}_b(t) \big]$ over the state space $\set{X}=(\set{X}_1 \cup \dots \cup \set{X}_{{B}})$.
The feasibility set of SBS $b$'s control at state $\vectx(t)$ is defined as $\set{Y}_b\big(t,\vectx(t)\big)=\{ {\sched}_{bm}\big(t,\vectx(t)\big)\in\{0,1\}, 0 \leq {\pow}_b\big(t,\vectx(t)\big) \leq \pow^{\texth{max}}\}$.
Note that we will frequently use $\vectY(t)$ instead $\vectY\big(t,\vectx(t)\big)$ for notational simplicity.
As the system evolves, UEs need to be scheduled at each time slot based on QSI and CSI.
The service quality of UE is ensured such that the average queue of packets intended for UE $m\in\set{M}_b$ at its serving SBS $b$ is finite, i.e. $\bar{q}_{bm}=\lim_{t\rightarrow\infty}\frac{1}{t}\sum_{\tau=0}^{t-1}q_{bm}(\tau) \leq \infty$.

The objective of this work is to determine the control policy per SBS $b$ which maximizes a utility function $f_b(\cdot)$ while ensuring UEs' quality of service (QoS).
The utility of a SBS at time $t$ is its EE which is defined as  $\one\tran\vect{\rate}_{b}\big(t,\vectY(t),\vect{\channel}(t)\big)/\big( p_b(t)+p_0 )$.
Here, $p_0$ is the fixed circuit power consumption at an SBS~\cite{jnl:shuguang05}.
Let $\bar\vectY=\lim_{t\rightarrow\infty}\frac{1}{t}\sum_{\tau=0}^{t-1}\vectY(\tau)$ be the limiting time average expectation of the control variables $\vectY(t)$.
Formally, the utility maximization problem for SBS $b$ can be written as:
\begin{subequations}\label{eqn:spectrum_sharing_optimization}
\begin{eqnarray}
	\label{eqn:spectrum_sharing_optimization_obj} &\underset{\bar\vecty_b}{\text{maximize}} & f_b(\bar\vecty_b,\bar\vecty_{-b}), \\
	\label{cns:user_QoS}& \text{subject to} & \bar{q}_{bm} \leq \infty \qquad \forall m\in\set{M}_b, \\
	\label{cns:collection}& & (\ref{eqn:evolution_queue}), (\ref{eqn:evolution_channels}), \\
	\label{cns:control_all}& & \vecty_b(t)\in\set{Y}_b(t,\vectx) \qquad \forall t.
\end{eqnarray}
\end{subequations}

Furthermore, we assume that SBSs serve their scheduled UEs for a time period of $T$.
Therefore, we use the notion of \emph{time scale separation} between transmit power allocation and UE scheduling processes, hereinafter.
For SBS $b\in\set{B}$, the transmit power allocation $\pow_b(t)$ is determined for each transmission and thus, is a \emph{fast} process.
However, UE scheduling $\vect{\sched}_b(t)$ is fixed for a duration of $T$ to ensure a stable transmission.
Therefore, UE scheduling is a \emph{slower} process than power allocation.

\subsection{Resource management as a dynamic stochastic game}\label{sec:MFG_stochastic}

We focus on finding a control policy which solves (\ref{eqn:spectrum_sharing_optimization}) over a time period $[0,T]$ for a given set of scheduled UEs considering the state transitions $\vectx(0)\rightarrow\vectx(T)$.
Therefore, we define the following time-and-state-based utility for SBS $b$:
\begin{align}
	\nonumber \util_b \big(0,\vectx(0)\big) &= \util_b \big(0,\vectx(0),\vectY(0)\big) \\
	&= \expect\big[ \int_0^T f_b\big(\tau,\vectx(\tau),\vectY(\tau)\big) d\tau \big].
\end{align}
The goal of each SBS is to maximize this utility over $\vecty_b(\tau) = \big[\vect{\sched}^\star_b(\tau),\pow_b(\tau)\big]$ for a given UE scheduling $\vect{\sched}^\star(\tau)$ subject to the system state dynamics $d\vectx(t) = \vect{X}_t dt + \vect{X}_z d\vect{z}(t)$, $\forall b\in\set{B}, \forall m\in\set{M}$ where: 
\begin{equation*}
\vect{X}_t= \Big[ \arrival_{bm}(t)-\rate_{bm}\big(t,\vectY(t),\vect{\channel}(t)\big), G\big( t, \channel_{bm}(t)\big) \Big]_{m\in\set{M},b\in\set{B}},
\end{equation*}
and $\vect{X}_z=\text{diag}(\zero_{\size{M}},\zeta\one)$.

As the network state evolves as a function of QSI and CSI, the strategies of the SBSs must adapt accordingly.
Thus, maximizing the utility $\util_b \big(0,\vectx(0)\big)$ for $b\in\set{B}$ under the evolution of the network states can be modeled as a dynamic stochastic game (DSG) as follows:
\begin{definition}
	For a given UE scheduling mechanism, the power control problem can be formulated as a dynamic stochastic game
		$\game = (\set{B}, \{\set{Y}_b\}_{b\in\set{B}}, \{\set{X}_b\}_{b\in\set{B}}, \{\util_b\}_{b\in\set{B}})$ where:
		\begin{itemize}
			\item $\set{B}$ is the set of players which are the SBSs.
			\item $\set{Y}_b$ is the set of actions of player $b\in\set{B}$ which encompass the choices of transmit power $\pow_b$ for given scheduled UEs $\vect{\sched}_b$.
			\item $\set{X}_b$ is the state space of player $b\in\set{B}$ consisting QSI $\vect{q}_b$ and CSI $\vect{h}_b$.
			\item $\util_b$ is the average utility of player $b\in\set{B}$ that depends on the state transition $\vectx(t)\rightarrow\vectx(T)$ as follows:
			\begin{equation}\label{eqn:running_utility}
			\util_b\big(t,\vectx(t)\big) = \textstyle\expect \big[ \int_t^T  f_b\big(\tau,\vectx(\tau),\vectY(\tau)\big) d\tau \big].
			\end{equation}
		\end{itemize}
\end{definition}
One suitable solution for the defined DSG is the closed-loop Nash equilibrium (CLNE) defined as follows:
\begin{definition}\label{def:CLNE}
	The control variables $\vectY^\star(t)=\big(\vecty^\star_b(t),\vecty^\star_{-b}(t)\big)\in\set{Y}\big(t,\vectx(t)\big)$ constitute a \emph{closed-loop Nash equilibrium} if, for $\set{Y}\big(t,\vectx(t)\big)=\Big( \set{Y}_1\big(t,\vectx(t)\big) \cup \dots \cup \set{Y}_{\size{B}}\big(t,\vectx(t)\big) \Big)$,
	\begin{multline}\label{eqn:CLNE}
		\textstyle\expect \big[ \int_t^T  f_b\big(\tau,\vectx(\tau),\vecty^\star_b(\tau),\vecty^\star_{-b}(\tau)\big) d\tau \big] 
		\\ \geq \textstyle\expect \big[ \int_t^T  f_b\big(\tau,\vectx(\tau),\vecty_b(\tau),\vecty^\star_{-b}(\tau)\big) d\tau \big],
	\end{multline}
	is satisfied $\forall\ b\in\set{B}$, $\forall\ \vectY(t)\in\set{Y}\big(t,\vectx(t)\big)$ and $\forall\ \vectx(t)\in\set{X}$.
\end{definition}
To satisfy (\ref{eqn:CLNE}), each SBS must have the knowledge of its own state $\vectx_b(t)$ and a feedback on the optimal strategy $\vecty^\star_{-b}(t)$ of the rest of SBSs and thus, the achieved equilibrium is known as a CLNE.
We denote by $\vect{\util}^\star\big(t,\vectx(t)\big) = \big[ \util_b^\star\big(t,\vectx(t)\big) \big]_{ b\in\set{B} }$ the trajectories of the utilities induced by the NE. 
The existence of the NE is ensured by the existence of a joint solution $\vect{\util}\big(t,\vectx(t)\big)=\big[ \util_b\big(t,\vectx(t)\big) \big]_{ \forall b\in\set{B} }$ to the following $\size{B}$ coupled HJB equations~\cite{book:gueant11}:
\begin{multline}\label{eqn:HJB_original}
	\textstyle \frac{\partial}{\partial {t}}[\util_b\big(t,\vectx(t)\big)] + \max\limits_{\vecty_b(t,\vectx(t))} \bigg[ \vect{X}_t \frac{\partial}{\partial {\vectx}}[\util_b\big(t,\vectx(t)\big)] 
	\\+ f_b(t,\vectx(t),\vectY(t)\big) 
	\textstyle  + \frac{1}{2}\text{tr}\Big( \vect{X}_z^2\frac{\partial^2}{\partial {\vectx}^2}[\util_b\big(t,\vectx(t)\big)] \Big) \bigg]  = 0,
\end{multline}
defined for each SBS $b\in\set{B}$ and tr($\cdot$) is the matrix trace operation.

\begin{proposition}\label{prop:ext_CLNE}
	(\emph{Existence of a closed-loop NE in $\game$}) A sufficient condition for the existence of a closed-loop NE in $\game$ is that for all $(v^\star, p^\star)$, we must have $2v^\star(p^\star+p_0) + \beta\ln(1+\beta p^\star) \neq 0$ where $\beta = \frac{|h_{bm}(t)^2|}{I_{bm}(t) + \noise}$ for all $b\in\set{B}$ and their scheduled UEs $m\in\set{M}$.
\end{proposition}

\begin{proof}
	Proving the existence of a solution to the HJB equation 
	 given in (\ref{eqn:HJB_original}) each SBS provides a sufficient condition for the existence of a NE for the DSG~\cite{book:gueant11}.
	Thus, there exists a solution if the function
	\begin{multline}
		\textstyle F\big(\vectx(t), \vecty_{-b}(t), \frac{\partial [\util_b\left(t,\vectx(t)\right)]}{\partial {q_b}} \big) = \max\limits_{\vecty_b(t,\vectx(t))} \bigg[ f_b((t,\vectx(t),\vectY(t)\big) 
		\\- r_{bm}(t,\vectx(t),\vectY(t) \big) \frac{\partial[\util_b(t,\vectx(t))]}{\partial {q_b}} \bigg],
	\end{multline} 
	is smooth~\cite{book:evans10}.
	According to (\ref{eqn:running_utility}), the average utility $\util_b\left(t,\vectx(t)\right)$ that depends on the state transition $\vectx(t)\to\vectx(T)$ is independent of the current choice of controls $\vecty_b(t)$.
	Therefore, the terms $\frac{\partial[\util_b(t,\vectx(t))]}{\partial {q_b}}$, $\frac{\partial [\util_b\left(t,\vectx(t)\right)]}{\partial {\vect{\channel}_b}}$ and $\frac{\partial^2 [\util_b\left(t,\vectx(t)\right)]}{\partial {\vect{\channel}^2_b}}$ become constant with respect to $\vecty_b(t)$ and thus, are neglected since they do not impact on $F(\cdot)$.
	Considering the first derivative, the maximizer needs to satisfy
	$ \frac{\partial f_b}{\partial {y\vecty_b}} - \frac{\partial r_b}{\partial {\vecty_b}} \frac{\partial \util_b}{\partial {q_b}} = 0$ which yields,
\begin{equation}\label{eqn:NE_derivative}
	 v (p + p_0)^2 =   \beta (p + p_0) - (1+\beta p)  \log(1+\beta p),
\end{equation}
	where $p$ is the transmit power of the respective SBS and $v = \frac{\partial \util_b}{\partial {q_b}} \beta$.
	Note that we have omitted the dependencies over time and states for notational simplicity.

	In order to evaluate the smoothness of $F\big(\vectx(t), \vecty_{-b}(t), \frac{\partial [\util_b\left(t,\vectx(t)\right)]}{\partial {q_b}} \big)$, we use the implicit function theorem~\cite{jnl:oliveira12}.
	Thus, we define a function $\mathbb{F}(v,p) = v (p + p_0)^2 - \beta (p + p_0) + (1+\beta p)  \log(1+\beta p)$ and check whether $\mathbb{F}(v,p) = 0$ is held.
	If $\mathbb{F}(v,p) = 0$ is held, there exists $\varphi:\Re \to \Re$ such that $p^\star = \varphi(v^\star)$ and $\frac{\partial [\varphi(v^\star)]}{\partial v} = {\frac{\partial [\mathbb{F}(v^\star,p^\star)]}{\partial v}}/{\frac{\partial [\mathbb{F}(v^\star,p^\star)]}{\partial p}} < \infty$.
	Since $\frac{\partial [\varphi(v^\star)]}{\partial v} = \frac{(p^\star+p_0)^2}{2v^\star(p^\star + p_0) + \beta \ln(1 + \beta p^\star)}$, the smoothness is ensured by 
	\begin{equation}\label{eqn:NE_sufficient_condition}
		2v^\star(p^\star + p_0) + \beta \ln(1 + \beta p^\star) \neq 0.
	\end{equation}
	
	Note that we consider $\frac{\partial \util_b}{\partial {q_b}} > 0$ and thus, $v>0$ and $v$ is bounded above, i.e. $0<v<\infty$.
	A bounded $v$ ensures the solution to (\ref{eqn:NE_derivative}) remains finite (finite transmit power) while $v>0$ guarantees  $2v^\star(p^\star + p_0) + \beta \ln(1 + \beta p^\star)>0$ for all feasible transmit powers. 
	Thus, the above sufficient condition (\ref{eqn:NE_sufficient_condition}) is held for all the particular choices of $v$ and $p$.
\end{proof}

Solving $\size{B}$ mutually coupled HJB equations is complex when $\size{B}>2$.
Furthermore, it requires gathering QSI and CSI from all the SBSs throughout the network which incurs a tremendous amount of information exchange.
Hence, it is impractical for UDNs with large $\size{B}$.
However, as $\size{B}$ grows extremely large, it can be observed that the set of all players transform into a continuum and their individuality becomes indistinguishable. 
In such a scenario, we can solve the dynamic stochastic game using the tools of mean-field theory, as explained next.

 \section{Optimal Power Control: A Mean-Field Approach}\label{sec:MFG}

As the number of SBSs becomes large ($\size{B}\rightarrow\infty$), we assume that the interference tends to be bounded in order to have non-zero rates as observed in \cite{pap:mari12,jnl:meriaux13} and \cite{book:couillet11}.
Moreover, each SBS implements a transmission policy based on the knowledge of its own state.
SBSs in such an environment are indistinguishable from one another thus, resulting in a continuum of players.
This is a reasonable assumption in UDNs because the impact of an individual SBS's decision on the system is negligible from the macroscopic view unless such a decision is adopted by a considerable large portion of the population.
As an example, the choice of a SBS to increase its energy consumption by 1 Watt in a system consisting of thousands of SBSs consuming couple of megawatts is insignificant.
However, it is a significant increment of energy consumption if the same decision is made by $10\%$ of the population together.
Remarkably, increasing the number of SBSs infinitely large and allow them to become a continuum
allows us to simplify the solution of all $\size{B}$ HJB equations by reducing them to two equations as discussed below.

At a given time and state $\big(t,\vectx(t)\big)$ and for a scheduled UE $m$, the impact of other SBSs on the choice of a given SBS $b\in\set{B}$ appears in the interference term:
\begin{equation*}
	I_{bm}\big(t,\vectx(t)\big) = \textstyle\sum_{\forall b'\in\set{B}\setminus\{b\}}\pow_{b'}\big(t,\vectx(t)\big)|\channel_{b'm}(t)|^2.
\end{equation*}
As the number of SBSs grows large, we consider that the interference is bounded as shown in \cite{book:couillet11,book:hekmat06,jnl:hekmat08} and \cite{jnl:haenggi09} since the path loss is modeled according to the inverse-square law.
Thus, the interference term can be rewritten by using a normalization factor for the channels~\cite{book:couillet11}.
Let $\eta/\size{B}$ be the normalization factor where $\eta$ is SBS density and thus, the channel gain becomes $\channel_{bm}(t) = \frac{\sqrt{\eta} \tilde{\channel}_{bm}(t)}{\sqrt{\size{B}}}$ with $\expect[|\tilde{\channel}_{bm}(t)|^2]=1$.
Thus, 
the interference can be rewritten as (\ref{eqn:long1})
\begin{figure*}
	\begin{equation}\label{eqn:long1}
	\begin{aligned}
	I_{bm}\big(t,\vectx(t)\big) &= \frac{\eta}{\size{B}} \sum_{\forall b'\in\set{B}\setminus\{b\}}\pow_{b'}\big(t,\vectx(t)\big)|\vectxx(t)|^2 
	= \frac{\eta}{\size{B}} \sum_{\forall b'\in\set{B}}\pow_{b'}\big(t,\vectx(t)\big)|\vectxx(t)|^2 - \frac{\eta}{\size{B}} \pow_{b}\big(t,\vectx(t)\big)|\tilde{\channel}_{bm}(t)|^2 \\
	&= {\eta} \int_{\set{X}}\pow_{b'}(t,\vectx_{b'})|\tilde{\channel}_{b'm}|^2 {\mass}^{\size{B}}(t,\vectx_{b'}) d\vectx_{b'}- \frac{\eta}{\size{B}} \pow_{b}\big(t,\vectx(t)\big)|\tilde{\channel}_{bm}(t)|^2.
	\end{aligned}
	\end{equation}
	\rule{\textwidth}{.5pt}
\end{figure*}
with a SBS state distribution $\vect{\mass}^{\size{B}}(t) = [ \mass^{\size{B}}\big(t,\vectx') ]_{\vectx'\in\set{X}}$ where
\begin{equation}\label{eqn:mass_distribution_finite}
\mass^{\size{B}}\big(t,\vectx') =  \frac{1}{\size{B}} \textstyle\sum\limits_{b=1}^{\size{B}} \delta\big(\vectx_b(t)=\vectx'\big).
\end{equation}
Here, $\delta(\cdot)$ is the \emph{Dirac delta} function and $\vectx'\in\set{X}$.

Now, we assume that each SBS has only the knowledge of its own state and, thus, it will implement a homogeneous transmission policy, i.e. $p_b\big(t,\vectx(t)\big) = \policy \big( t, \vectx_b(t) \big)$ for all $b\in\set{B}$, based the knowledge that it has available.
Furthermore, the transmission policy satisfies $\expect [\int_0^T \policy^2\big( t, \vectx_b(t) \big) dt]\allowbreak<\infty$.
With the above assumptions, the SBS states $\vectSubTwoTime{\vectx}{\size{B}}{t}$ for all $t$ become exchangeable under the transmission policy $\policy$~\cite{book:koch82,book:gagliardini14},
i.e.,
\begin{equation}
\set{L}\big(\vectSubTwoTime{\vectx}{\size{B}}{t}\big) = \set{L}\big(\vectx_{\pi(1)},\dots,\vectx_{\pi(\size{B})}\big) \qquad \text{for }\policy,\forall t,
\end{equation} 
where $\set{L}(\cdot)$ is any joint rule of the SBS states and $\pi(\cdot)$ produces any permutation over $\{\seta{\size{B}}\}$.
Due to this exchangeability, the game $\game$ now consists of a set of generic players including a generic SBS with the state $\vectxx(t)$ at time $t$.
Thus, the interference under the homogeneous transmission policy $\policy$ with respect to the state distribution $\vect{\mass}^{\size{B}}(t)$ will be given by:
\begin{multline}\label{eqn:mf_intef_interm}
	I\big(t,\vect{\mass}^{\size{B}}(t)\big) = I_{bm}\big(t,\vectx(t)\big) = {\eta} \int_{\set{X}}\policy\tNx|\tilde{\channel}|^2 \mass^{\size{B}}\tNx dx 
	\\- \frac{\eta}{\size{B}} \policy\big(t,\vectx(t)\big)|\tilde{\channel}_{bm}(t)|^2.
\end{multline}
As $\size{B}\to\infty$, the set of players becomes a continuum. 
Since the transmission policy $\policy\big(t,\vectxx(t)\big)$ preserves the indistinguishable property (exchangeability in SBS states), there exists a limiting distribution $\vect{\mass}(t)$ such that~\cite{pap:tembine11,jnl:meriaux13},
\begin{equation}
	\vect{\mass}(t) = \lim_{\size{B} \to \infty} \vect{\mass}^{\size{B}}(t).
\end{equation}
This limiting distribution of the states is defined as the MF.

\begin{figure}[t]
	\centering
	\includegraphics[width=.9\columnwidth]{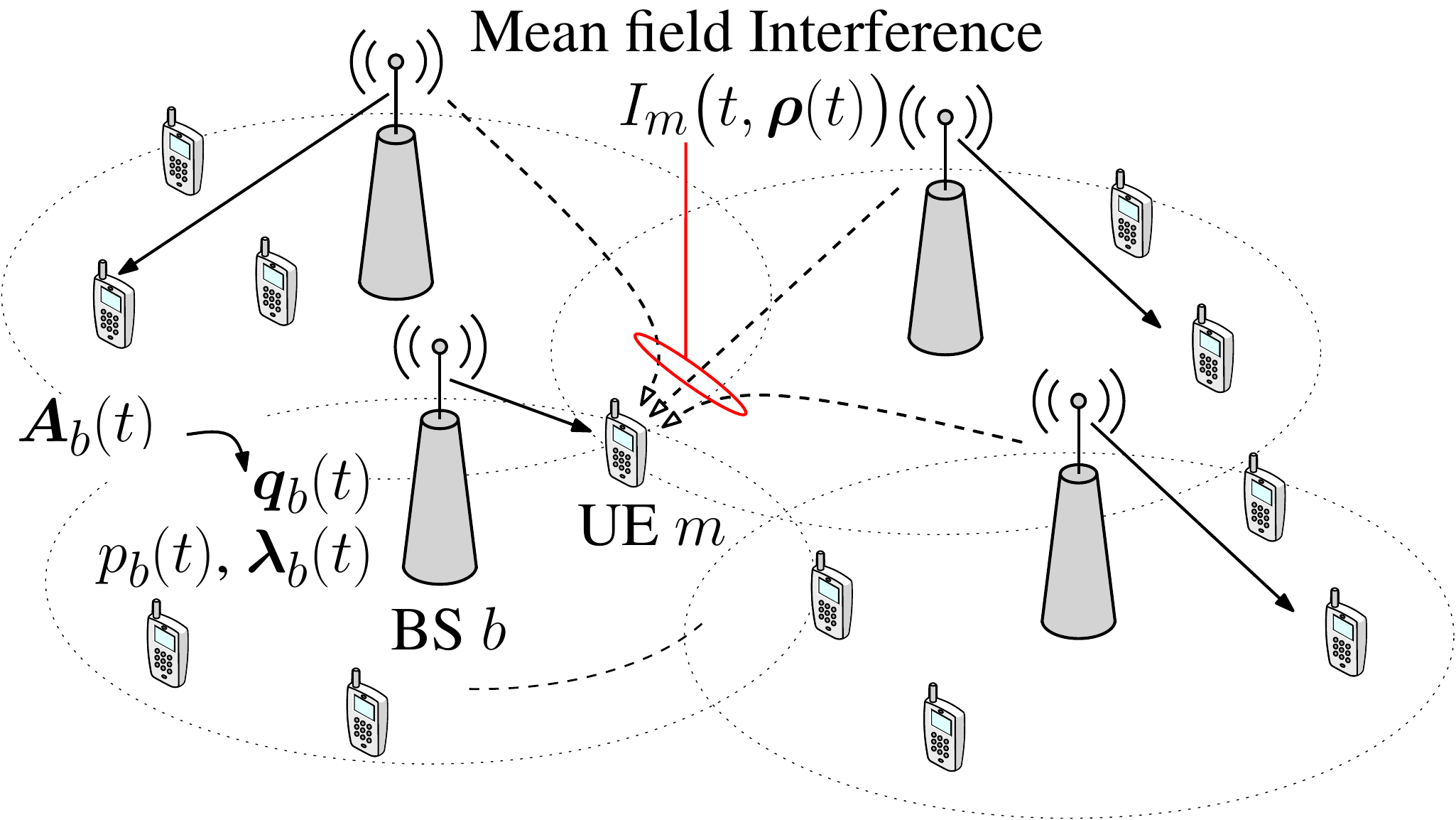}
	\caption{Mean-field interference from the perspective of an SBS. As the number of interfering SBSs grows large, the interference observed at a generic UE becomes independent of individual states and transmission policies of SBSs and only depends on the time $t$ and the limiting distribution $\vect{\mass}(t)$.}
	\label{fig:system_model}
\end{figure}

Furthermore, when $\expect[|\tilde{\channel}_{bm}(t)|^2]=1$, the channel gain $|\tilde{\channel}_{bm}(t)|^2$ is a finite quantity for any given time $t$, i.e. $|\tilde{\channel}_{bm}(t)|^2<\infty$.
Since the homogeneous transmission policy satisfies $\expect [\int_0^T \policy^2\big( t, \vectx(t) \big) dt]<\infty$ and is bounded by a maximum transmit power, we have $\policy\big(t,\vectx(t)\big)<\infty$. 
Therefore, 
\begin{multline*}
\policy\big(t,\vectx(t)\big)<\infty, |\tilde{\channel}_{bm}(t)|^2<\infty  \\
\begin{aligned}
	 & \implies \policy\big(t,\vectx(t)\big)|\tilde{\channel}_{bm}(t)|^2<\infty, \\
	 & \implies \frac{\eta}{\size{B}} \policy\big(t,\vectx(t)\big)|\tilde{\channel}_{bm}(t)|^2 < \infty, \\
	 &\implies \lim\limits_{\size{B}\to\infty } \frac{\eta}{\size{B}} \policy\big(t,\vectx(t)\big)|\tilde{\channel}_{bm}(t)|^2 = 0.
\end{aligned}
\end{multline*} 
Applying this result on (\ref{eqn:mf_intef_interm}),
as $\size{B}\to\infty$, the interference
\begin{equation}\label{eqn:mean_field_interference}
	I\big(t,\vect{\rho}(t)\big) = {\eta} \int_{\set{X}} \pow\big(t,\vectx'\big)|\tilde{\channel}|^2\rho(t,\vectx') d\vectx',
\end{equation}
converges to a finite value which depends on $\vect{\mass}^{\size{B}}(t)$.
With the convergence of the interference along the indistinguishable transmission policy and SBS states, the DSG $\game$ converges to an MFG~\cite{pap:tembine11}.
A simplified illustration of the system is shown in Fig.~\ref{fig:system_model}.

Thus, the utility maximization problem and the evolution of the states are reformulated for a generic SBS as follows:
\begin{subequations}\label{eqn:utility_maximization_MFG}
\begin{eqnarray}
	&\underset{ \big(\pow(t)|\sched^\star(t)\big), \forall t\in[0,T] }{\text{maximize}} & \util \big(0,\vectx(0)\big), \\
	&\text{subject to} & dx(t) = X_t dt + X_z d\vect{z}(t), \\
	& & \vecty\big(t,\vectxx(t)\big)\in\set{Y}\big(t,\vectx(t)\big),
\end{eqnarray}
\end{subequations}
for all $t\in[0,T]$
where 
	${X}_t=\big[D\big(t,y(t)\big), G\big( t, \tilde{\channel}(t)\big) \big]$ and
a diagonal matrix ${X}_z=\text{diag}(0,\zeta\one)$. 
Here, $D\big(t,\vecty(t)\big)=\arrival(t) - \rate(t,\vecty(t),\tilde{\channel}(t),\vect{\rho}(t)\big)$.
Similarly, the utility of a generic SBS $\util\big(t,\vectxx(t)\big)$ follows (\ref{eqn:running_utility}) with the necessary modifications.
The formal definition of the MF equilibrium is as follows:
\begin{definition}\label{def:MFE}
	The control vector $\vecty^\star=(\vect{\sched}^\star,\pow^\star)\in\set{Y}\big(t,\vectx(t)\big)$ constitutes a \emph{mean-field equilibrium} if, for all $\vecty\in\set{Y}\big(t,\vectx(t)\big)$ with the MF distribution $\vect{\rho}^\star$, it holds that,
\begin{equation*}
	\util(\vecty^\star,\vect{\rho}^\star) \geq \util(\vecty,\vect{\rho}^\star).
\end{equation*}
\end{definition}
In the MF framework, the MF equilibrium given by the solution $\big[\util^\star\big(t,\vectxx(t)\big) , \rho^\star\big(t,\vectxx(t)\big)\big]$ of (\ref{eqn:utility_maximization_MFG}) is equivalent to the NE of the $\size{B}$-players DSG~\cite{book:gueant11}.
Moreover, the optimal trajectory $\util^\star\big(t,\vectxx(t)\big)$ is found by applying backward induction to a single HJB equation and 
the MF (limiting distribution) $\rho^\star\big(t,\vectxx(t)\big)$ is obtained by forward solving the FPK equation as (\ref{eqn:MFG_PDEs}),
\begin{figure*}
	\begin{equation}\label{eqn:MFG_PDEs}
	\begin{cases}
	\frac{\partial}{\partial {t}}[\util\big(t,\vectxx(t)\big)] + \max\limits_{\pow(t,\vectxx(t))} \bigg[ D\big(t,\vecty^\star(t)\big) \frac{\partial}{\partial {q}}[\util\big(t,\vectxx(t)\big)] + f(t) + \Big( G( t, \tilde{\channel}) \frac{\partial}{\partial {\tilde{\channel}}}  + \frac{\zeta^2}{2} \frac{\partial^2}{\partial {\tilde{\channel}}^2} \Big) \util\big(t,\vectxx(t)\big) \bigg] = 0, \\
	\frac{\partial}{\partial {\tilde{\channel}}}\Big[G(t,\tilde{\channel})\rho\big(t,\vectxx(t)\big)\Big]  - \frac{\zeta^2}{2} \frac{\partial^2}{\partial {\tilde{\channel}}^2}[\rho\big(t,\vectxx(t)\big)]  + \frac{\partial}{\partial {q}}\Big[ D\big(t,\vecty^\star(t)\big) \rho\big(t,\vectxx(t)\big) \Big] + \partial_{t}\rho\big(t,\vectxx(t)\big)  = 0,
	\end{cases}
	\end{equation}
	\rule{\textwidth}{.5pt}
\end{figure*}
respectively.
The optimal transmit power strategy is given by,
\begin{multline}\label{eqn:optimal_strategy_MFG}
	\textstyle \pow^\star\big(t,\vectxx(t)\big) = \argmax\limits_{\pow(t,\vectxx(t))} \bigg[ X_t \frac{\partial}{\partial x}[\util\big(t,\vectxx(t)\big)] + f(t) 
	\\ \textstyle  + \frac{1}{2}\text{tr}\Big( X_z^2\frac{\partial^2}{\partial x^2}[\util\big(t,\vectxx(t)\big)] \Big) \bigg].
\end{multline}
Solving (\ref{eqn:MFG_PDEs})-(\ref{eqn:optimal_strategy_MFG}) yields the behavior of a generic SBS in terms of transmission power, utility and state distribution.

The main motivation behind this mean-field analysis is to reduce the complexity of the $B$-player DSG.
To solve the original DSG with $B$ SBSs/players presented in Section \ref{sec:MFG_stochastic}, one must solve $B$ HJB PDEs given by (\ref{eqn:HJB_original}).
However, the solution of the MFG with large $B$ can be found by simply solving the two coupled HJB-FPK PDEs rather solving $B$ coupled HJB PDEs as in the original problem. 
Clearly, this reduces complexity considerably and can scale well with the number of players.
The relation between the solution of $B$-player DSG ($\vectY^\star$) and the solution of MFG ($\vectY^\star_0$) is $\vectY^\star = \vect{Y}^\star_0 + \sum_{n=1}^\infty \frac{\vectY_n}{B^n}$ where $\vectY_n$ is the $n$-th order correction coefficient~\cite{book:gueant11}.
It should be noted that the solution of the MFG is equivalent to the solution of the $B$-player DSG only for large $B$.
Therefore, it could be challenging to solve the HJB-FPK coupled PDEs for small number of SBSs and thus, the proposed method is not applicable for highly sparse networks.
As the number of SBSs increases, the solution of the DSG tends towards the solution of the MFG.
Therefore, for dense to ultra dense networks, the complexity of the proposed method is much lower compared to solving the original $B$-player DSG.
In fact, the complexity of the proposed method remains constant for all large $B$.
Fig.~\ref{fig:converg_plot_man} illustrate the complexity of the proposed method by the number of iterations needed to solve the HJB-FPK coupled PDEs for different number of SBSs.
\begin{figure}[!t]
	\centering
	\includegraphics[width=\myfigfactorx\columnwidth]{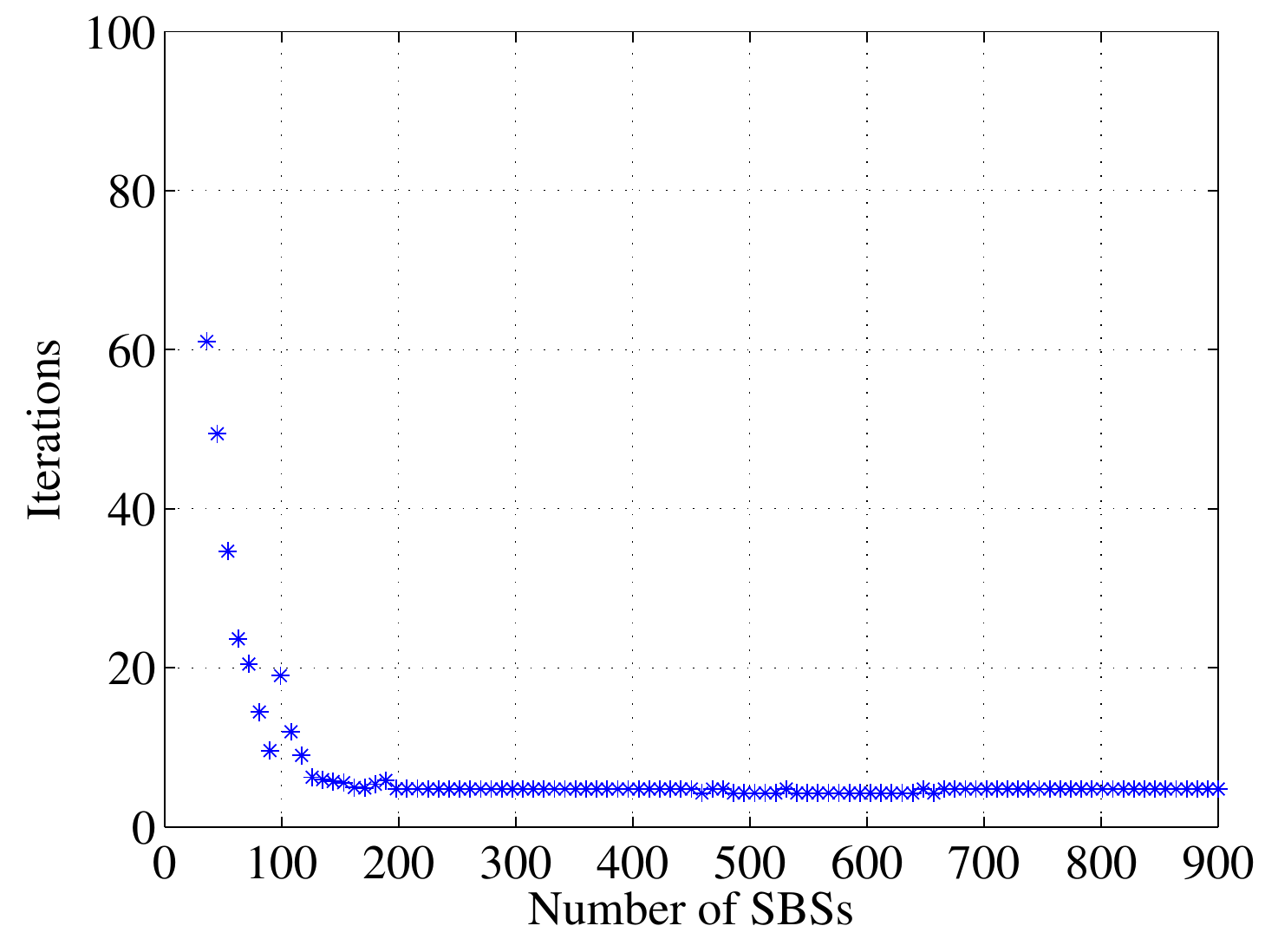}
	\caption{Number of iterations to solve the HJB-FPK PDEs for different number of SBSs in the system.
	The plot consists of a set of results for different number of SBSs  over a fixed area of $0.5625~\text{km}^2$}
	\label{fig:converg_plot_man}
\end{figure}
Here, it can be observe that for systems with large number of SBSs, the MFE is achieved by a small number of iterations and thus, the complexity remains fixed.

\section{UE Scheduling Via Lyapunov Framework}\label{sec:formulations}

\begin{figure}[t]
	\centering
	\includegraphics[width=.9\columnwidth]{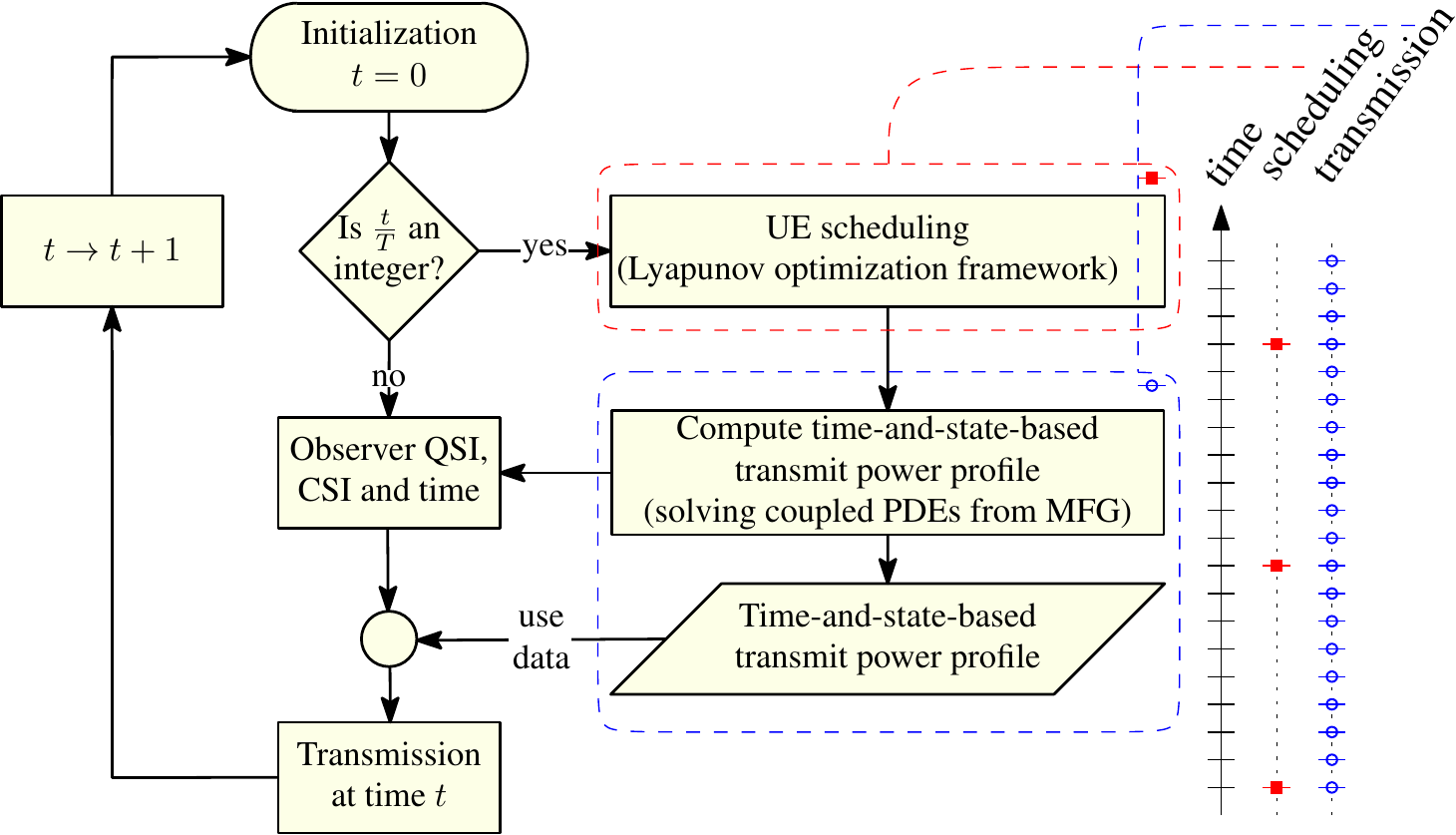}
	\caption{Inter-relation between the MFG and the Lyapunov framework.}
	\label{fig:flow_chart}
\end{figure}

By using  time scale separation, the scheduling variables are decoupled from the mean-field game and thus, can be optimized separately.
Some of the existing baseline approaches for UE scheduling are proportional fair (PF) scheduling in terms of rates, best-CSI based UE scheduling, and scheduling based on highest QSI~\cite{jnl:huang11,jnl:neely10,jnl:kwan09} and \cite{jnl:fritzsche15}.
PF scheduling ensures the fairness among UEs in terms of their average rates (history) while the latter methods exploit the instantaneous CSI or QSI.
Using such conventional schedulers for solving (\ref{eqn:spectrum_sharing_optimization}) within a UDN context will not properly account for the inherent CSI and QSI dynamics over space and time, thus yielding poor performance.
Therefore, we solve the original utility maximization problem of SBS $b\in\set{B}$ with respect to the scheduling variables as follows;
\begin{subequations}\label{eqn:scheduling_optimization}
\begin{eqnarray}
	&\underset{\big( \bar{\vect{\sched}}_b|\pow^\star,\vect{\rho}^\star \big)}{\text{maximize}} & f_b(\bar\vecty_b,\bar\vecty_{-b}), \\
	\label{cns:schedule_constraints_all}&\text{subject to} & (\ref{cns:user_QoS}), (\ref{cns:collection}), \\
	\label{cns:schedules}& & \vect{\sched}_b(t) \in\set{L}\big(t,\vectx(t)\big) \qquad \forall t.
\end{eqnarray}
\end{subequations}

Here, the feasible set $\set{L}\big(t,\vectx(t)\big)$ consists of all the vectors with $\sched_{bm}(t)\in[0,1]$ and $\one\tran\vect{\sched}_b(t) = 1$.
Note that the scheduling variables are relaxed from integers to real numbers for the ease of analysis.
Furthermore, we consider UE $m\in\set{M}_b$ is expected to be served by SBS $b$ with a data rate of $\hat{\rate}_{bm}(t)=\sched_{bm}(t) \tilde{\rate}_{bm}(t)$.
\begin{proposition}
	Given that the transmit power profile $\policy^\star\big(t,\vectxx(t)\big)$ and the state distribution $\vect{\mass}^\star(t)$, the expected data rate of UE $m\in\set{M}_b$ at time $t$ is calculated as follows:
	\begin{equation}\label{eqn:datarate_ue_expected}
		\tilde{\rate}_{bm}(t) = \omega \log_2 \bigg( 1 + \frac{  \policy^\star\big(t,\vectx_b(t)\big)|\tilde{h}_{bm}(t)|^2 }{ I\big(t,\vect{\rho}^\star(t)\big) + \noise } \bigg),
	\end{equation}
	assuming that UE $m$ is scheduled by SBS $b$.
\end{proposition}
\begin{proof}
	As discussed in Section~\ref{sec:MFG}, as $\size{B}\to\infty$, all SBSs will follow the optimal transmission policy given in (\ref{eqn:optimal_strategy_MFG}).
	Moreover, solving (\ref{eqn:MFG_PDEs}) yields the MF $\vect{\mass}^\star(t)$ and thus, the expected interference becomes $I\big(t,\vect{\rho}^\star(t)\big)$.
	For any given time $t$, since all the SBSs aware of their own states, the knowledge of $\vect{q}_b(t)$ and $h_{bm}(t)$ is available at SBS $b$.
	Therefore, as UE $m$ is scheduled by SBS $b$, the data rate given in (\ref{eqn:datarate_ue}) becomes (\ref{eqn:datarate_ue_expected}).
\end{proof}

In order to solve the stochastic optimization problem (\ref{eqn:scheduling_optimization}) per SBS $b$, the \emph{drift plus penalty} (DPP) approach in Lyapunov optimization framework can be applied~\cite{book:neely10,book:georgiadis06}.
The Lyapunov DPP approach decomposes the stochastic optimization problem into sub-policies that can be implemented in a distributed way.
Therefore, $\size{B}$ copies of problem (\ref{eqn:scheduling_optimization}) are locally solved at each SBS and, thus, the proposed solution can apply for a large number of SBSs i.e., as $\size{B}\rightarrow\infty$.

First, a vector of auxiliary variables $\vect{\auxSched}_b(t)=\big[ \auxSched_{bm}(t) \big]_{ m\in\set{M}_b }$
is defined to satisfy the constraints (\ref{cns:schedules}).
These additional variables are chosen from a set $\set{V}=\{\vect{\auxSched} | \vect{\auxSched}\tran\one=1~\text{and}~\zero\preceq\vect{\auxSched}\preceq\one \}$ independent from both time and state.
Thus, (\ref{eqn:scheduling_optimization}) is transformed as follows;
\begin{subequations}\label{eqn:scheduling_opt_equivalent}
	\begin{eqnarray}
	&\underset{\bar{\vect{\sched}}_b,\bar{\vect{\auxSched}}_b}{\text{maximize}} & f_b(\bar\vecty_b,\bar\vecty_{-b}), \\
	\label{cns:schedule_eq_constraints_all}&\text{subject to} &  (\ref{cns:schedule_constraints_all}), (\ref{cns:schedules}), \\
	\label{cns:aux_equivalence}& & \bar{\vect{\auxSched}}_b = \bar{\vect{\sched}}_b, \\
	\label{cns:aux_schedules}& & \vect{\auxSched}_b(t) \in\set{V} \qquad \forall t.
	\end{eqnarray}
\end{subequations}
To satisfy the constraint in (\ref{cns:aux_equivalence}), we introduce a set of virtual queues $\vqSched_{bm}(t)$ for each associated UE $m\in\set{M}_b$. 
The evolution of virtual queues will follow~\cite{book:neely10};
\begin{equation}
\label{eqn:virtual_queue_scheduling} \vqSched_{bm}(t+1) = \vqSched_{bm}(t) + \auxSched_{bm}(t) - \sched_{bm}(t).
\end{equation}
Consider the combined queue  $\Xi_b(t)=\big[ \vect{q}_b(t),\vect{\vqSched}_b(t) \big]$
and its quadratic Lyapunov function $L\big(\Xi_b(t)\big) = \frac{1}{2}\Xi_b\tran(t)\Xi_b(t)$.
Modifying (\ref{eqn:evolution_queue}) considering a chunk of time, the evolution of the queue of UE $m\in\set{M}_b$ can be reformulated as $q_{bm}(t+1) = \max \big(0, q_{bm}(t) + \arrival_{bm}(t) - \hat{\rate}_{bm}(t)\big)$.
Thus, one-slot drift of Lyapunov function $\Delta L = L\big(\Xi_b(t+1)\big)-L\big(\Xi_b(t)\big)$ is given by,
\begin{equation*}
\Delta L \!\! = \!\! \textstyle\frac{ \big( \vect{q}_b\tran(t+1)\vect{q}_b(t+1) - \vect{q}_b\tran(t)\vect{q}_b(t) \big)
	+ \big( \vect{\vqSched}_b\tran(t+1)\vect{\vqSched}_b(t+1) - \vect{\vqSched}_b\tran(t)\vect{\vqSched}_b(t) \big)
}{2}.
\end{equation*}
Neglecting the indices $b,~m$ and $t$ for simplicity and using,
\begin{eqnarray*}
	&([q + \arrival - \hat{\rate}]^+)^2 &\leq q^2 + (\arrival-\hat{\rate})^2 + 2 q(\arrival-\hat{\rate}), \\
	&(\vqSched + \auxSched - \sched)^2 &\leq \vqSched^2 + (\auxSched-\sched)^2 + 2 \vqSched(\auxSched-\sched),
\end{eqnarray*}
the one-slot drift can be simplified as follows:
\begin{equation}\label{eqn:drift_bound}
	\Delta L \leq K
	+ \vect{q}_b\tran(t) \big( \vect{\arrival}_b(t)-\vect{\hat{\rate}}_b(t) \big)
	+ \vect{\vqSched}_b\tran(t) \big( \vect{\auxSched}_b(t) - \vect{\sched}_b(t) \big),
\end{equation}
where $K$ is a uniform bound on the term 
$\big( \vect{\arrival}_b(t)-\vect{\hat{\rate}}_b(t)\big)\tran\big( \vect{\arrival}_b(t)-\vect{\hat{\rate}}_b(t)\big) 
+ \big( \vect{\auxSched}_b(t)-\vect{\sched}_b(t)\big)\tran\big( \vect{\auxSched}_b(t)-\vect{\sched}_b(t)\big)$.
The conditional expected Lyapunov drift at time $t$ is defined as $\Delta\big(\Xi_b(t)\big) = \expect[ L\big(\Xi_b(t+1)\big) -  L\big(\Xi_b(t)\big) |  \Xi_b(t) ]$.
Let $V\leq 0$ be a parameter which control the tradeoff between queue length and the accuracy of the optimal solution of (\ref{eqn:scheduling_opt_equivalent}) and $\vect{\sched}_b^{\texth{avg}}(t)=\frac{1}{t}\sum_{\tau=0}^{t-1}\vect{\sched}_b(\tau)$
 be the current running time averages of scheduling
variables.
Introducing a penalty term 
$V \nabla_{\vect{\sched}_b}\tran f\big(\vect{\sched}_b^{\texth{avg}}(t)\big)\allowbreak \expect [ \big(\vect{\sched}_b(t)\big) | \Xi_b(t)]$
 to the expected drift and minimizing the upper bound of the drift DPP,
\begin{multline*}
	K
	+ V \nabla_{\vect{\sched}_b}\tran f\big(\vect{\sched}_b^{\texth{avg}}(t)\big) \expect [ \big(\vect{\sched}_b(t)\big) | \Xi_b(t)]
	\\+ \expect [ \vect{q}_b\tran(t) \big( \vect{\channel}_b(t)-\vect{\hat{\rate}}_b(t) \big) | \Xi_b(t)]
	\\+ \expect [ \vect{\vqSched}_b\tran(t) \big( \vect{\auxSched}_b(t) - \vect{\sched}_b(t) \big) | \Xi_b(t)],
\end{multline*}
yields the control policy of SBS $b$.
Thus, the objective of SBS $b$ is to minimize the below expression given by,
\begin{multline}\label{eqn:join_objective}
\Big[  
\overbrace{\vect{\vqSched}_b\tran(t)\vect{\auxSched}_b(t)}^{\parbox{7em}{\centering \footnotesize Impact of virtual queue and auxiliaries}}
\Big]_{\#1} +
\Big[ 
\overbrace{V \nabla_{\vect{\sched}_b}\tran f\big(\vect{\sched}_b^{\texth{avg}}(t)\big) \vect{\sched}_b(t)}^{\text{penalty}}
\\- \underbrace{\vect{q}_b\tran(t)\vect{\hat{\rate}}_b(t)}_{\text{QSI and CSI}} -
\underbrace{\vect{\vqSched}_b\tran(t)\vect{\sched}_b(t) }_{\parbox{7em}{\centering\footnotesize Impact of virtual queue and scheduling}}
 \Big]_{\#2},
\end{multline}
at each time $t$.
The terms $K$ and $\vect{q}_b\tran(t)\vect{\arrival}_b(t)$
are neglected since they do not depend on $\vect{\sched}_b(t)$ and $\vect{\auxSched}_b(t)$.
Note that terms $\#1$ and $\#2$ have decoupled the scheduling variables and the auxiliary variables, respectively.
Thus, the subproblems of finding auxiliary variables and scheduling variables can be presented as below. 

\subsection{Evaluation of auxiliary variables}

Since the auxiliary variables are decoupled from scheduling variables as given in (\ref{eqn:join_objective}), the formal representation of auxiliary variables evaluation at SBS $b$ for time $t$ is as follows:
\begin{subequations}\label{eqn:sub_prob_auxiliary}
	\begin{eqnarray}
	&\underset{\vect{\nu}}{\text{minimize}} & \vect{\vqSched}_b\tran(t)\vect{\nu}, \\
	&\text{subject to} & \vect{\nu} \in\set{V}.
	\end{eqnarray}
\end{subequations}
It can be noted that the feasible set is a convex hull.
Due to the affine nature of the objective function, we can conclude that the optimal solution for (\ref{eqn:sub_prob_auxiliary}) should lie on a vertex of the convex hull~\cite{book:boyd}.
Thus, the optimal solution is given by,
\begin{equation}\label{eqn:sub_prob_auxiliary_sol}
	{\auxSched}^\star_{bm}(t) = 
	\begin{cases}
		1, & \text{if}~m = \argmin_{m\in\set{M}_b} \Big( \vqSched_{bm}(t)\Big), \\
		0, & \text{otherwise}.
	\end{cases}
\end{equation}

\subsection{Determining the scheduling variables}

\begin{algorithm}[!t]
	\caption{UE Scheduling Algorithm Per SBS}
	\label{alg:scheduling}
	\begin{algorithmic}[1]                    
		\STATE {\bf Input:} $\vect{q}_b(t)$ and $\vect{\vqSched}_b(t)$ for $t=0$ and SBS $b\in\set{B}$.
		\WHILE{ true }
		\STATE Observation: queues $\vect{q}_b(t)$ and $\vect{\vqSched}_b(t)$, and running averages $\vect{\sched}_b^{\texth{avg}}(t)$.
		\STATE Auxiliary variables: $\vect{\auxSched}_b(t) = \argmin_{\vect{\nu}\in\set{V}} \vect{\vqSched}_b\tran(t)\vect{\nu}$.
		\STATE Scheduling: $\vect{\sched}_b(t) = \argmax_{\vect{\delta}\in\set{L}(t,\vectx)}  
		\vect{q}_b\tran(t)\vect{\hat{\rate}}_b(t)
		+ \vect{\vqSched}_b\tran(t)\vect{\delta}
		- V \nabla_{\vect{\delta}}\tran f\big(\vect{\sched}_b^{\texth{avg}}(t)\big)\vect{\delta}$.
		\STATE Update: $\vect{q}_b(t+T)$, $\vect{\vqSched}_b(t+T)$ and $\vect{\sched}_b^{\texth{avg}}(t+T)$.
		\STATE $t\rightarrow t+T$
		\ENDWHILE
	\end{algorithmic}
\end{algorithm}

Optimal scheduling of SBS $b$ at time $t$ is found from solving a subproblem based on term $\#2$ in (\ref{eqn:join_objective}) as follows:

\begin{subequations}\label{eqn:sub_prob_scheduling}
	\begin{eqnarray}
	&\underset{\vect{\delta}}{\text{max}} & \vect{q}_b\tran(t)\vect{\hat{\rate}}_b(t)
	+ \vect{\vqSched}_b\tran(t)\vect{\delta}
	- V \nabla_{\vect{\delta}}\tran f\big(\vect{\sched}_b^{\texth{avg}}(t)\big)\vect{\delta}, \\
	&\text{s.t.} & \vect{\delta} \in\set{L}\big( t,\vectx(t)\big).
	\end{eqnarray}
\end{subequations}
Since the feasible set $\set{L}\big( t,\vectx(t)\big)$ is a convex hull and the objective is an affine function of $\vect{\delta}$, the optimal solution of (\ref{eqn:sub_prob_scheduling}) lies on a vertex of the convex hull $\set{L}\big( t,\vectx(t)\big)$.
Thus, the optimal scheduling at time $t$ is given by,
\begin{equation}\label{eqn:sub_prob_scheduling_sol}
{\sched}^\star_{bm}(t) = 
\begin{cases}
1, & \text{if}~m = m^\star, \\
0, & \text{otherwise},
\end{cases}
\end{equation}
where $m^\star=\argmax_{m\in\set{M}_b} \Big( q_{bm}(t)\tilde{r}_{bm}(t) + \vqSched_{bm}(t) - V \frac{\partial}{\partial{\sched_{bm}}}[f\big(\vect{\sched}_b^{\texth{avg}}(t)\big)] \Big)$.

\begin{remark}
	Even though we relax the boolean scheduling variables to continuous variables in the set $\set{L}\big( t,\vectx(t)\big)$, the optimal solution given in (\ref{eqn:sub_prob_scheduling_sol}) results a boolean scheduling vector.
	Thus, we claim that relaxing (\ref{eqn:scheduling_optimization}) does not change the optimality of the scheduling.
\end{remark}

The Lyapunov DPP method ensures that the gap between the time average penalty and the optimal solution is bounded by the term $\frac{K}{|V|}$~\cite{book:georgiadis06,jnl:neely10}, and \cite{jnl:bethanabhotla13}.
Thus, the optimality of the solution is ensured by the choice of a sufficiently large $|V|$.
The interrelation between the MFG and the Lyapunov optimization is illustrated in Fig.~\ref{fig:flow_chart} and UE scheduling algorithm which solves (\ref{eqn:scheduling_opt_equivalent}) is given in Algorithm \ref{alg:scheduling}.

\section{Numerical Results}\label{sec:results}

For simulations, the dimensions of the problem must be simplified in order to solve the coupled PDEs using a finite element method.
We used the MATLAB PDEPE solver for this purpose.
The utility of an SBS at time $t$ is its EE $\one\tran\vect{r}(t)/\big( p(t)+p_0 )$ and the ultimate goal is to maximize the average expected energy efficiency defined as $\expect\big[\int_{0}^{T_0}\one\tran\vect{r}(t)/\big( p(t)+p_0 ) dt \big]$.
Here, $T_0$ is the duration of the entire simulation and $p_0$ is the fixed circuit power consumption at an SBS~\cite{jnl:shuguang05}.
We assume that channels are not time-varying and thus, the state is solely defined by the QSI\footnote{
	Note that the assumption of fixed channels for simulations does not impact the theoretical contribution.
	Furthermore, even under non-fading channels, the provided results can model the ergodic behavior of the system.
}.
The initial 
limiting distribution, $\vect{\rho}\big(0,\vectx(0)\big)$, are assumed to follow a 
Gaussian distributions with means $0.5$ and variance $0.1$.
The choice of the final utility, the boundary condition,  $\Gamma\big(T,\vectx(T)\big)=-4\exp\big(\vectx(T)\big)$ is to encourage the scheduled UE to obtain an almost empty queue by the end of its scheduled period $T$.
The arrival rate $A(t)$ for a UE is modeled as a Poisson process with a mean of $\bar{A}=200~\text{kbps}$.
For simulations, the QSI is normalized by the queuing capacity of an SBS which is assumed to be $Q_{\max}=10\bar{A}$ kb.
Moreover, for simulations we assume that the transmit power is $\pow\in[0,1]~\text{Watts}$, the circuit power consumption is $\pow_0=1~\text{Watt}$, and the variance of Gaussian noise is $\noise=-70~\text{dBm}$.
The system parameters and channel models are based on \cite{online:3gpp10}.

The proposed method is compared to a baseline method which uses a proportional fair UE scheduler and an adaptive transmission policy.
Similar to the proposed method, SBSs in the baseline method seek to maximize EE by optimizing the transmit powers.
However, in the baseline, SBSs are not able to optimize their control parameters over the states due to the fact that they are oblivious to the states of the rest of the network.
Thus, the SBSs in the baseline method adopt a myopic approach in which they seek to maximize the \emph{instantaneous EE} $\one\tran\vect{r}(t)/\big( p(t)+p_0 )$ by optimizing $\pow(t)$ for each $t$ subject to UEs' instantaneous QoS requirement.
Moreover, each SBS needs to estimate the interference in order to solve its optimization problem.
Here, we assume that SBSs use the time average interference from their past experience and consider it as a valid estimation for the interference.
The SBS density of the system is defined by the average inter-site-distance (ISD) normalized by half of the minimum ISD, i.e. 1 unit of ISD implies an average distance of $20~\text{m}$.
The average load per SBS is
 $k=\frac{M}{B}$ while during the discussion on results, we use \emph{low load} for $k=2$ UEs per SBS and \emph{high load} for $k=5$ UEs per SBS.
Once a UE is scheduled, 100 transmissions will take place within the time period of $T=1$.

\subsection{Mean-field equilibrium of the proposed model}

\begin{figure}[!t]
	\centering
	\subfloat[MF distribution at the equlibrium as a function of time and QSI.]{
		\includegraphics[width=\myfigfactorx\columnwidth]{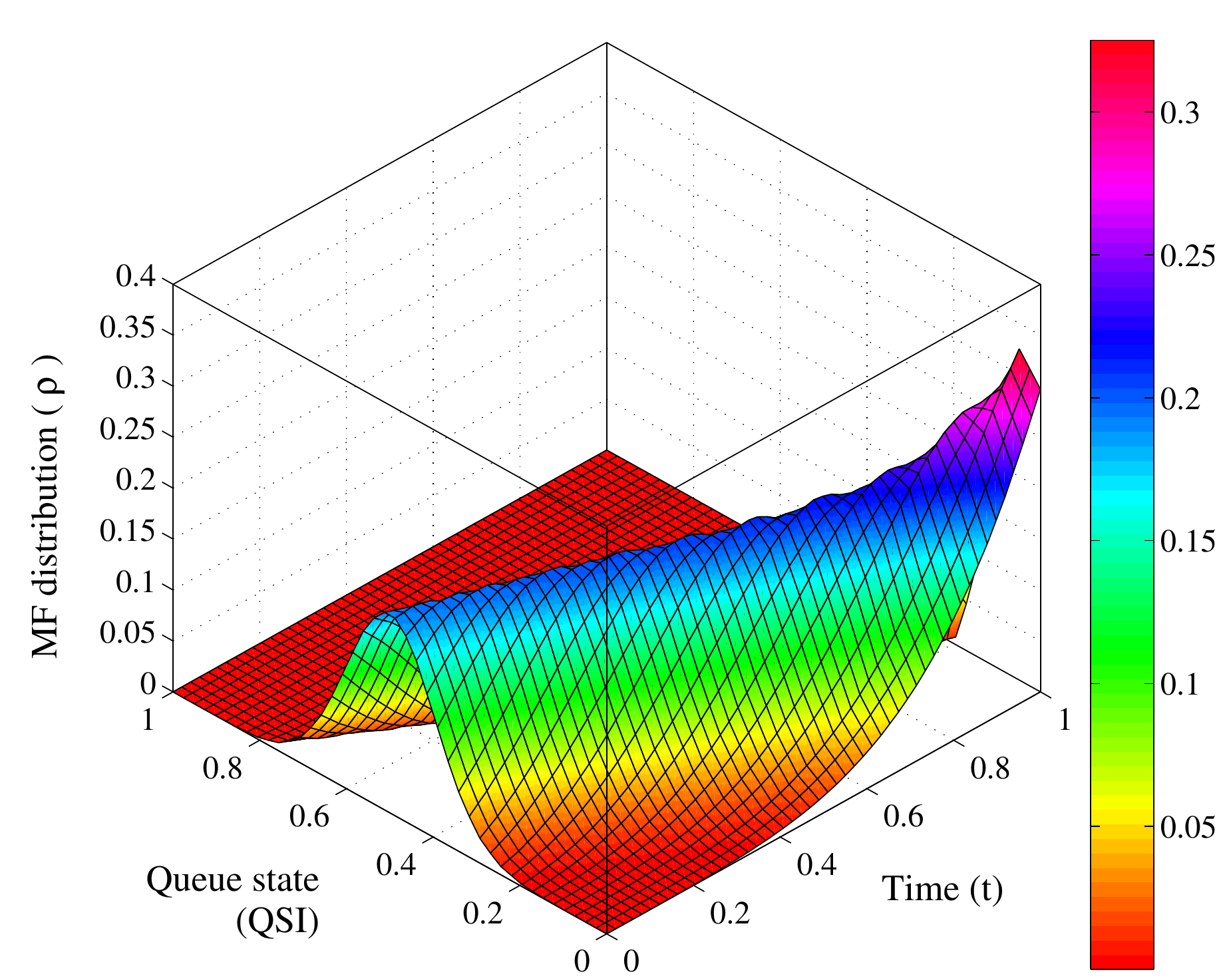}
		\label{fig:mf_distribution3D}
	}
	\hfil
	\subfloat[Evolution of the limiting distribution for a given QSI.]{
		\includegraphics[width=\myfigfactorx\columnwidth]{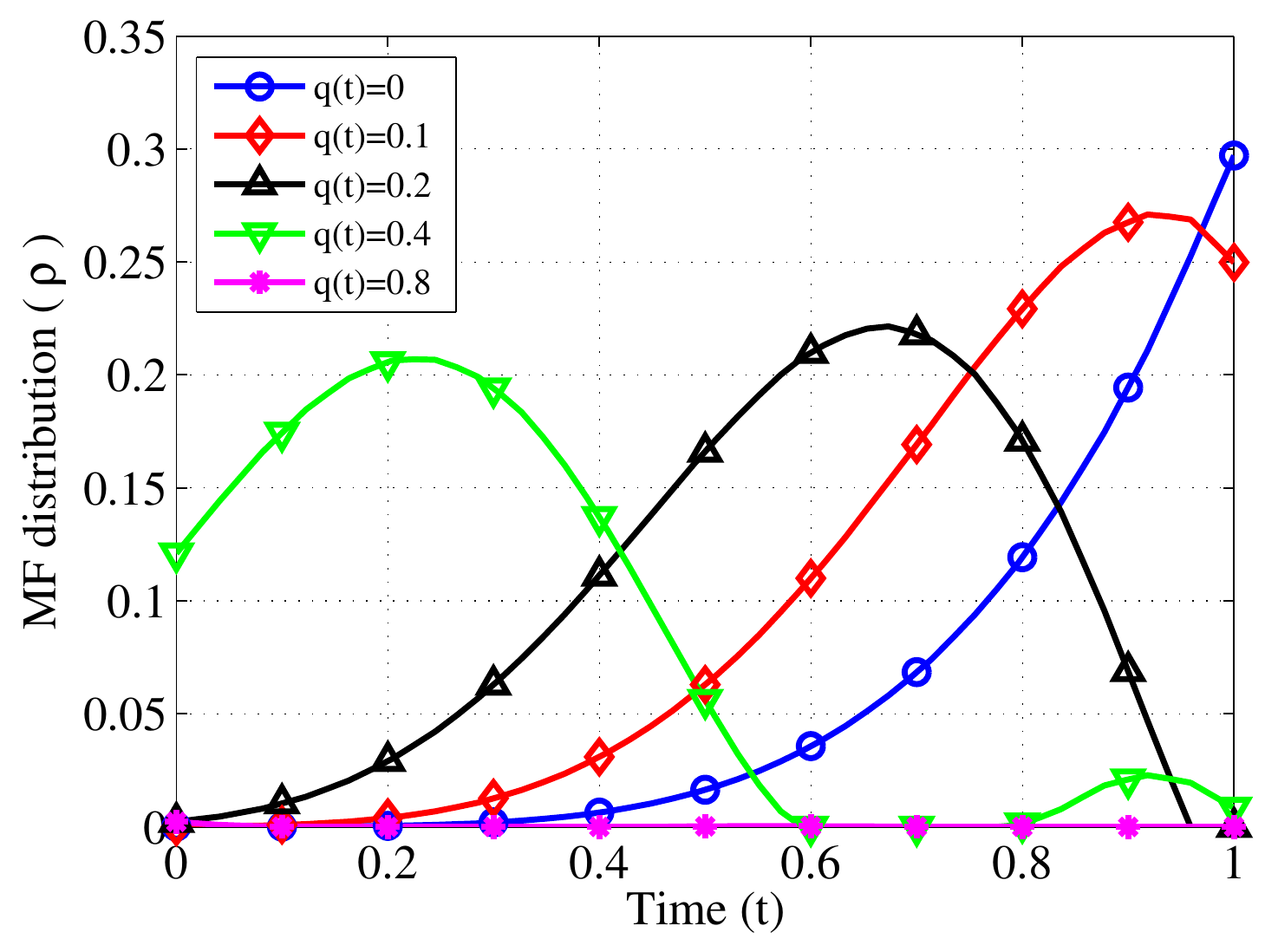}
		\label{fig:mf_distribution2D}
	}
	\caption{Evolution of limiting distribution ${\mass}^\star(t,q)$ at the MF equilibrium.}
	\label{fig:mean_field}
\end{figure}

Fig.~\ref{fig:mean_field} shows the MF distribution $\vect{\rho}^\star(t,q)$, which represents the evolution of the QSI distribution of scheduled UEs over time at the MF equilibrium.
During the period of $T=1$, SBSs transmit to their scheduled UEs and
expect to achieve a QSI close to zero by $t=T$ as shown in Fig.~\ref{fig:mf_distribution3D}.
Noted that, by the end of the transmission phase, the number of scheduled UEs with high QSI decreases thus, allowing the SBSs to schedule a new set of UEs by the next UE scheduling phase.
Here, the queues for all UEs will buildup whether they are scheduled or not due to the continuous arrivals with mean $\bar{A}$ at their respective serving SBSs.
As the queue grows beyond an SBS's capacity $Q_{\max}$ (i.e. $\text{QSI}>1$), the arrivals are dropped and thus, the associated UE will suffers from an outage.
Therefore, achieving a QSI almost close to zero for a scheduled UE provides the opportunity to build up its queue over a large period thus contributing to a reduction in the outage probability.
In Fig.~\ref{fig:mf_distribution2D}, we can see that the fraction of queues with $q(t)=\{0.2,0.4,0.8\}$ vanishes before the transmission duration ends.
As time evolves, the queues get empty based on the rates 
prior to new arrivals
and thus, an oscillation is observed for the queue fractions with $q(t)=0.4$ while $q(t)=0$ exhibits a monotonic increase.

\begin{figure}[!t]
	\centering
	\subfloat[Transmit power at the MF equilibrium as a function of time and QSI.]{
		\includegraphics[width=\myfigfactorx\columnwidth]{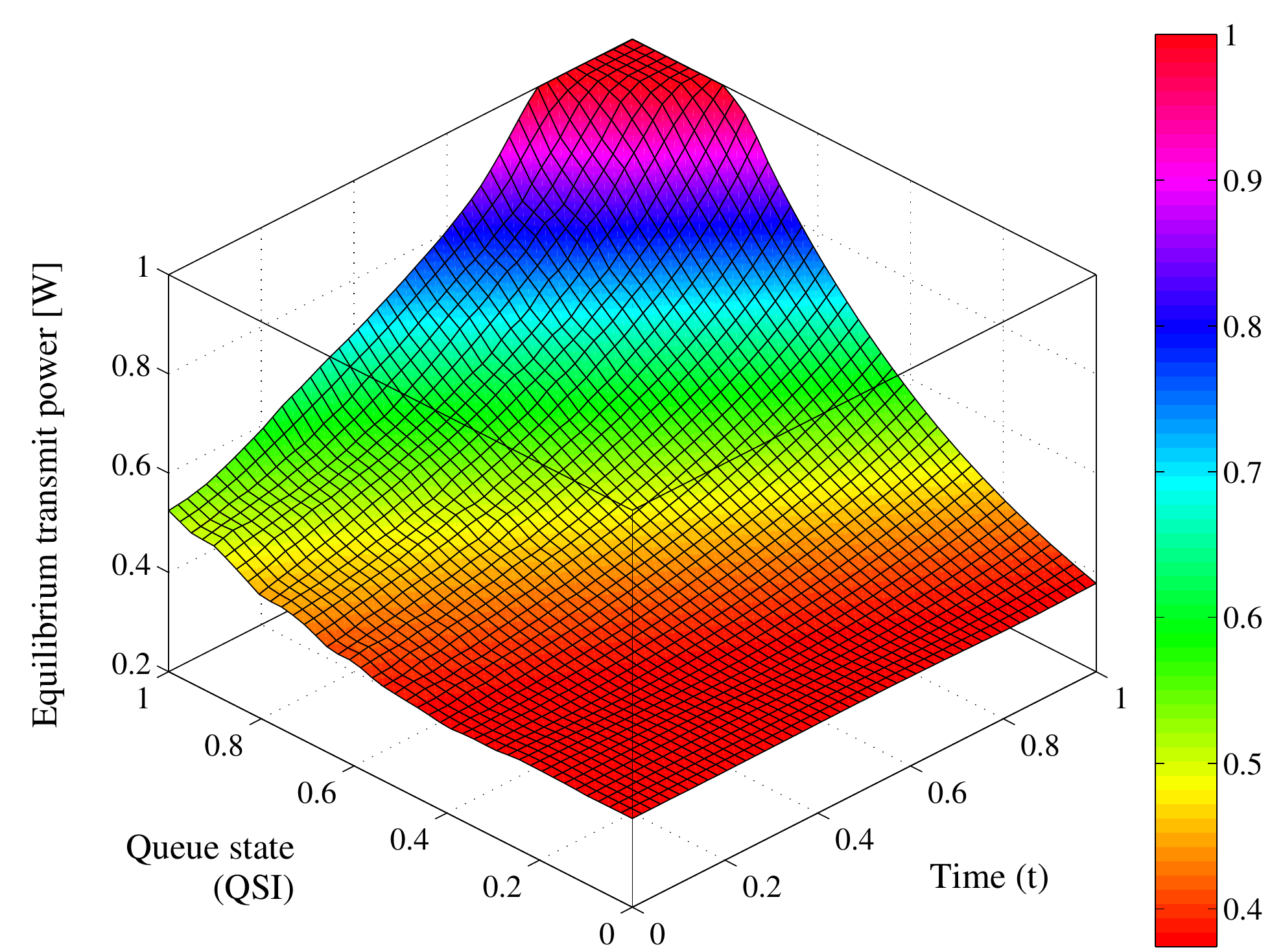}
		\label{fig:power_distribution3D}
	}
	\hfil
	\subfloat[Evolution of the transmit power for a given QSI.]{
		\includegraphics[width=\myfigfactorx\columnwidth]{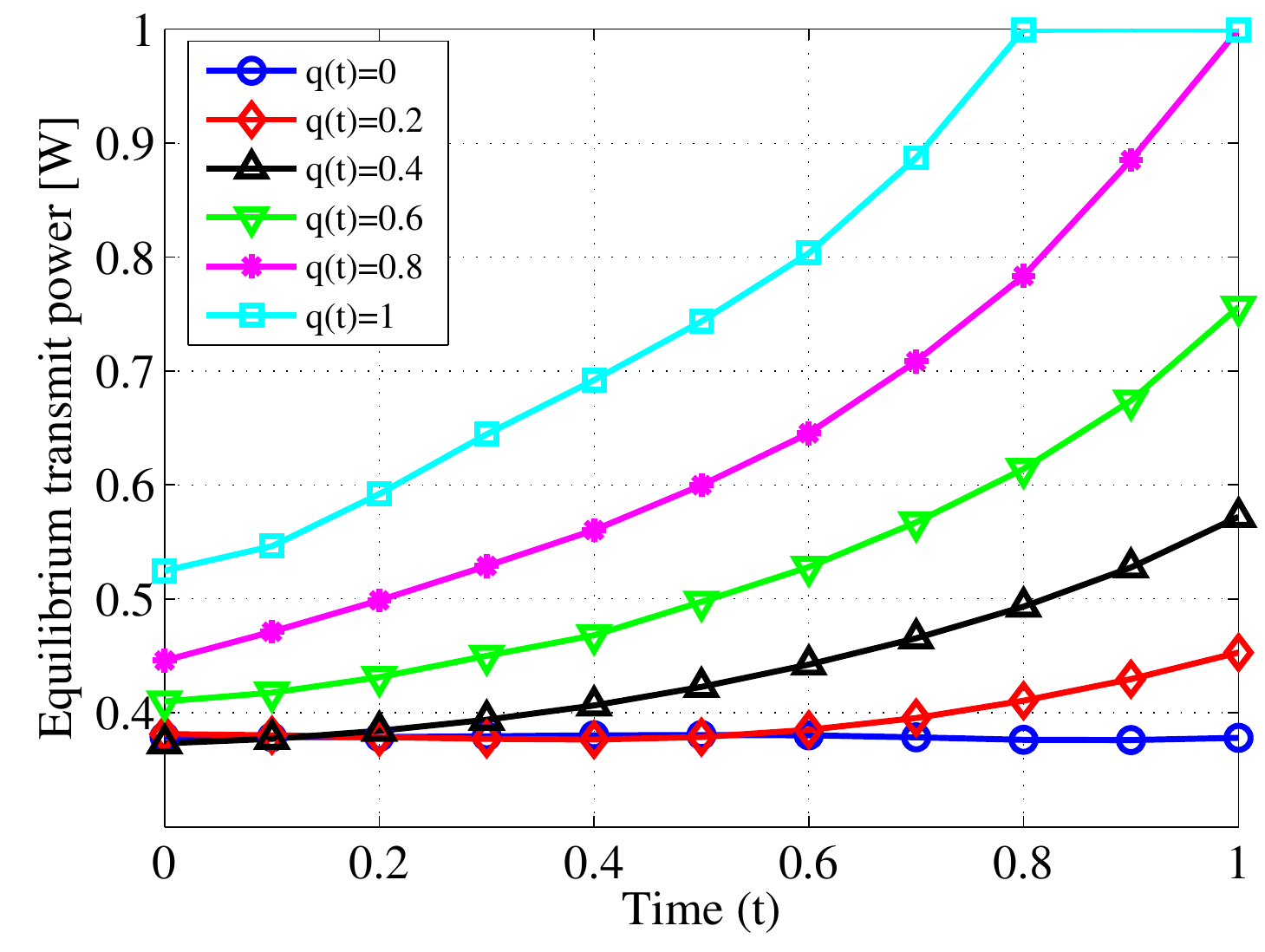}
		\label{fig:power_distribution2D}
	}
	\caption{Evolution of transmit power ${\pow}^\star(t,q)$ at the MF equilibrium.}
	\label{fig:power_distribution}
\end{figure}

The transmit power policy at the MF equilibrium is shown in Fig.~\ref{fig:power_distribution}.
It can be observed that a higher transmit power is needed when the QSI is high and this power can be reduced at low QSI, showing the overall EE of the proposed approach.
At $t=0$, a moderate transmit power is used even for UEs with high QSI.
Thus, SBSs prevent unnecessary interference within the system.
As the scheduling period arrives to an end, we recall that the choice of $\Gamma\big(T,\vectx(T)\big)=-4\exp\big(\vectx(T)\big)$ forces SBSs to obtain smaller QSI at $t=T$.
For UEs with low QSI, SBSs use moderate transmit power to maximize EE.
However, for UEs with high QSI, SBSs use higher transmit power to provide high data rates in order to empty their respective queues as well as preventing outages.
Thus, as time evolves, SBSs increase their transmit power for UEs with high QSI as illustrated in Fig.~\ref{fig:power_distribution2D} thereby improving the final utility.

\subsection{Energy efficiency and outage comparisons}\label{subsec:ee_op_comp}

\begin{figure}[!t]
	\centering
	\subfloat[Comparison of EE in terms of transmit bits per unit energy for low and high loads $k=\{2,5\}$.]{
		\includegraphics[width=\myfigfactorx\columnwidth]{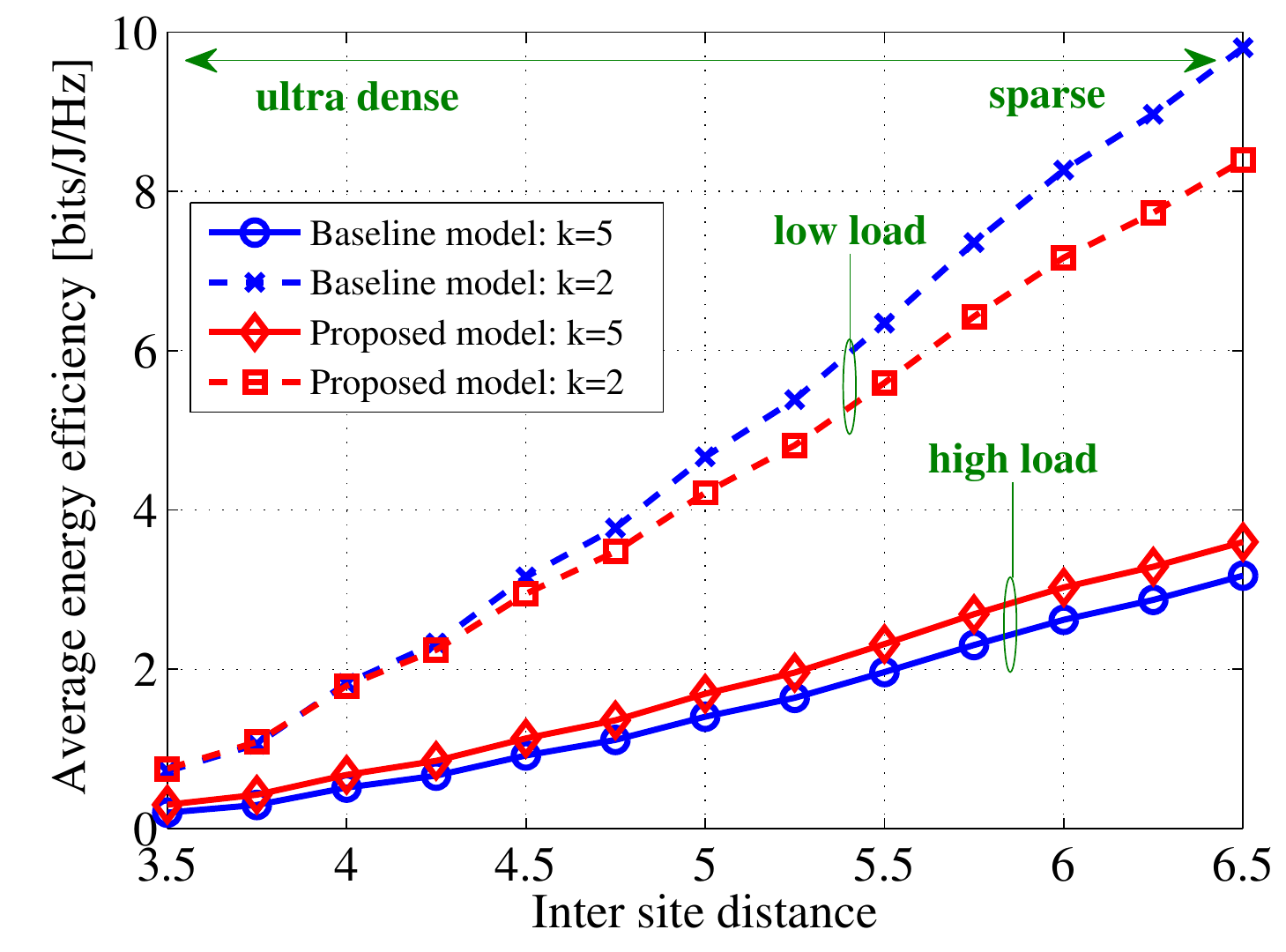}
		\label{fig:bpe_bsChange}
	}
	\hfil
	\subfloat[Comparison of outage probabilities for low and high loads $k=\{2,5\}$.]{
		\includegraphics[width=\myfigfactorx\columnwidth]{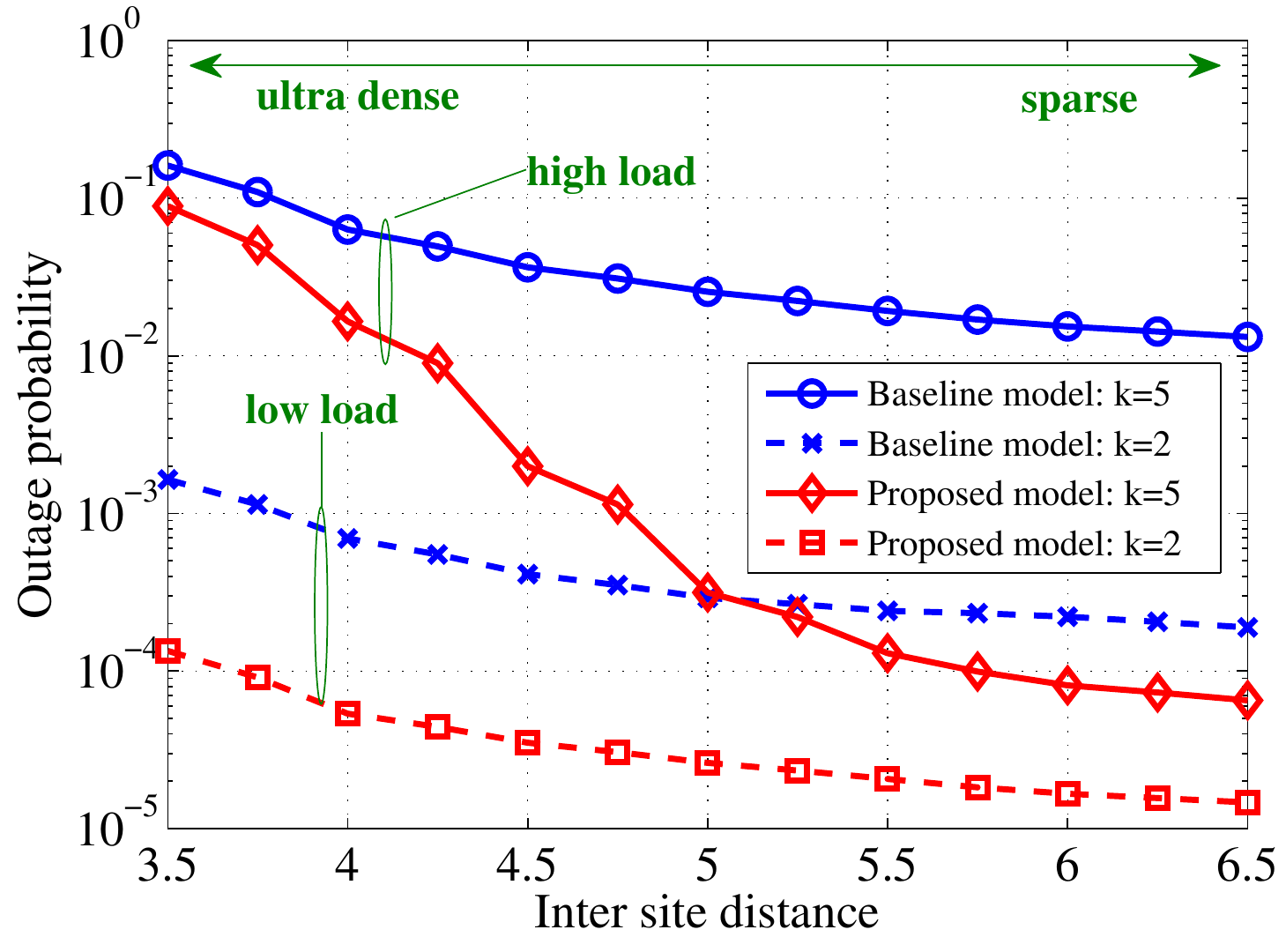}
		\label{fig:outage_bsChange}
	}
	\caption{Comparison of the behavior of EE and outage probabilities for different SBS densities.}
	\label{fig:bs_change}
\end{figure}

In Fig.~\ref{fig:bpe_bsChange}, we show the EE of the system in terms of the number of bits transmitted per joule of energy consumed as a function of ISD.
Here, the load represents the number of UEs served by an SBS.
For a dense network ($\isd = 3.5$), the proposed method will improve EE of about $5\%$ compared to the baseline model with $k=2$.
However, as the load increases to $k=5$, this EE improvement will reach up to $48.8\%$.
These improved EE for the dense scenario shown in the proposed approach is due to its ability to adapt to the dynamics of the network, as enabled by the stochastic game formulation and its MF approximation.
Naturally, as the network becomes less dense, the advantages of the proposed approach will be smaller.
For instance, for a non-dense network with an $\isd= 6.5$ and $k=5$ UEs per SBS, performance advantages of the MF approach decreases to $4.8\%$ as shown in Fig.~\ref{fig:bpe_bsChange}.
For a very lightly loaded and non-dense network, such as when $\isd=6.5$ and $k=2$ UEs per SBS, the use of the baseline method will be slightly more advantageous when compared to the proposed approach. 
This advantage is of about $14.1\%$ of improvement in the EE.

Fig.~\ref{fig:outage_bsChange} illustrates the outage probability as a function of ISD.
Here, the outage probability is defined as the fraction of unsatisfied UEs whose arrivals are dropped due to limitations of the queue capacity.
It can be noted that the proposed method uses smart UE scheduling and transmit power policy based on QSI which ensures a higher UE QoS.
Although the baseline model optimizes the transmit power with the goal of improving EE, it is unable to track the QSI dynamics and adapt the UE scheduling accordingly.
Thus, higher number of UEs suffer from outages.
As the network becomes highly loaded and more dense (i.e., as the number of SBSs and UEs increase), the number of UEs in outage will naturally increase for both approaches.
This is due to the increased interference and the high average waiting time due to increased number of UEs.  
From Fig.~\ref{fig:outage_bsChange}, we observe that the proposed method yields $91.8\%$ and $41.8\%$ reductions in outages compared to the baseline model for both low and high loads, respectively, in UDNs.
For a sparse network, the outage reductions are $92.2\%$ and $99.5\%$ for low and high loads, respectively.

\begin{figure}[!t]
	\centering
	\subfloat[Comparison of EE for sparse and dense networks $\text{ISD}=\{5.75,3.5\}$.]{
		\includegraphics[width=\myfigfactorx\columnwidth]{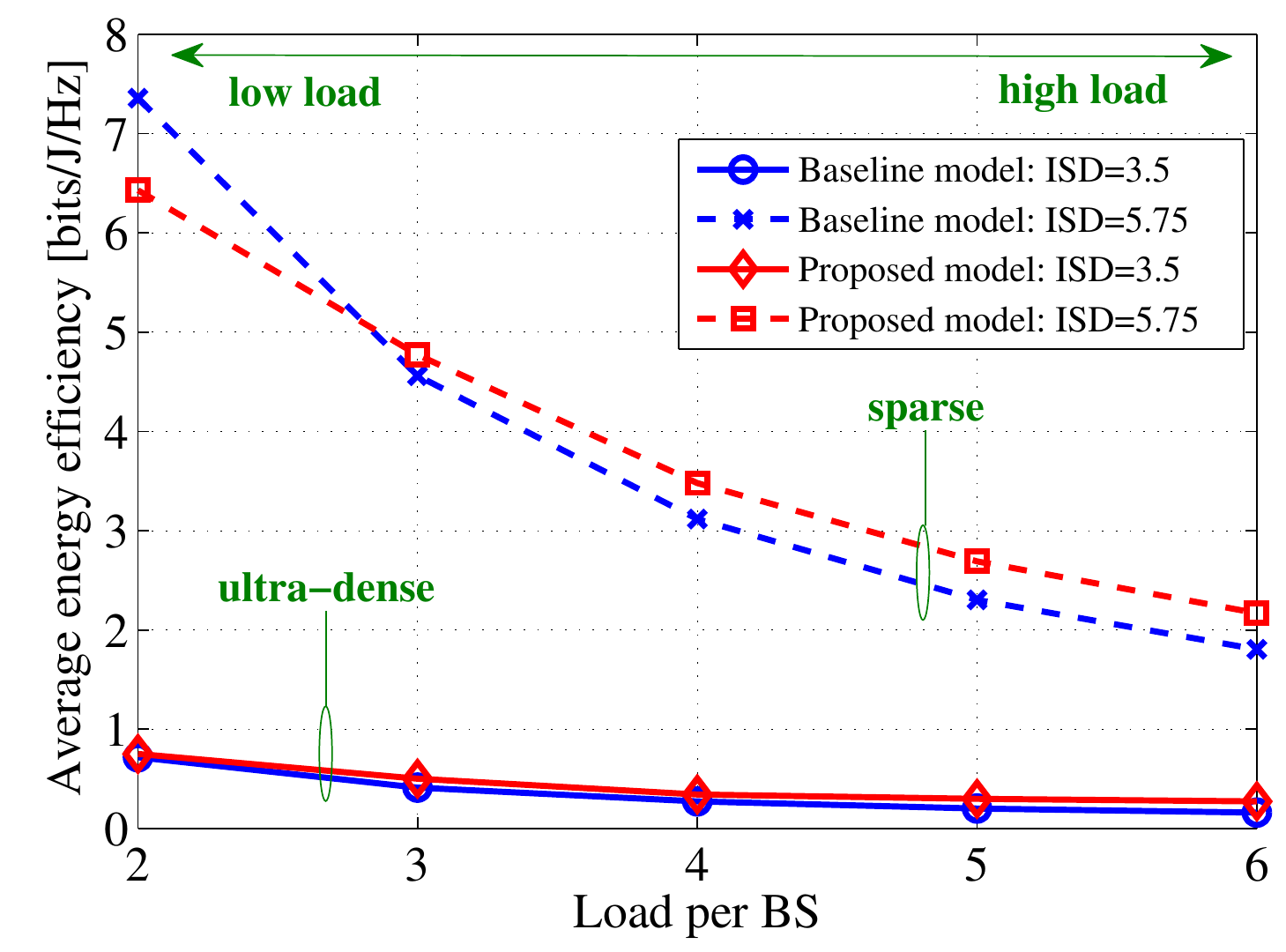}
		\label{fig:bpe_ueChange}
	}
	\hfil
	\subfloat[Comparison of outage probability for sparse and dense networks $\text{ISD}=\{5.75,3.5\}$.]{
		\includegraphics[width=\myfigfactorx\columnwidth]{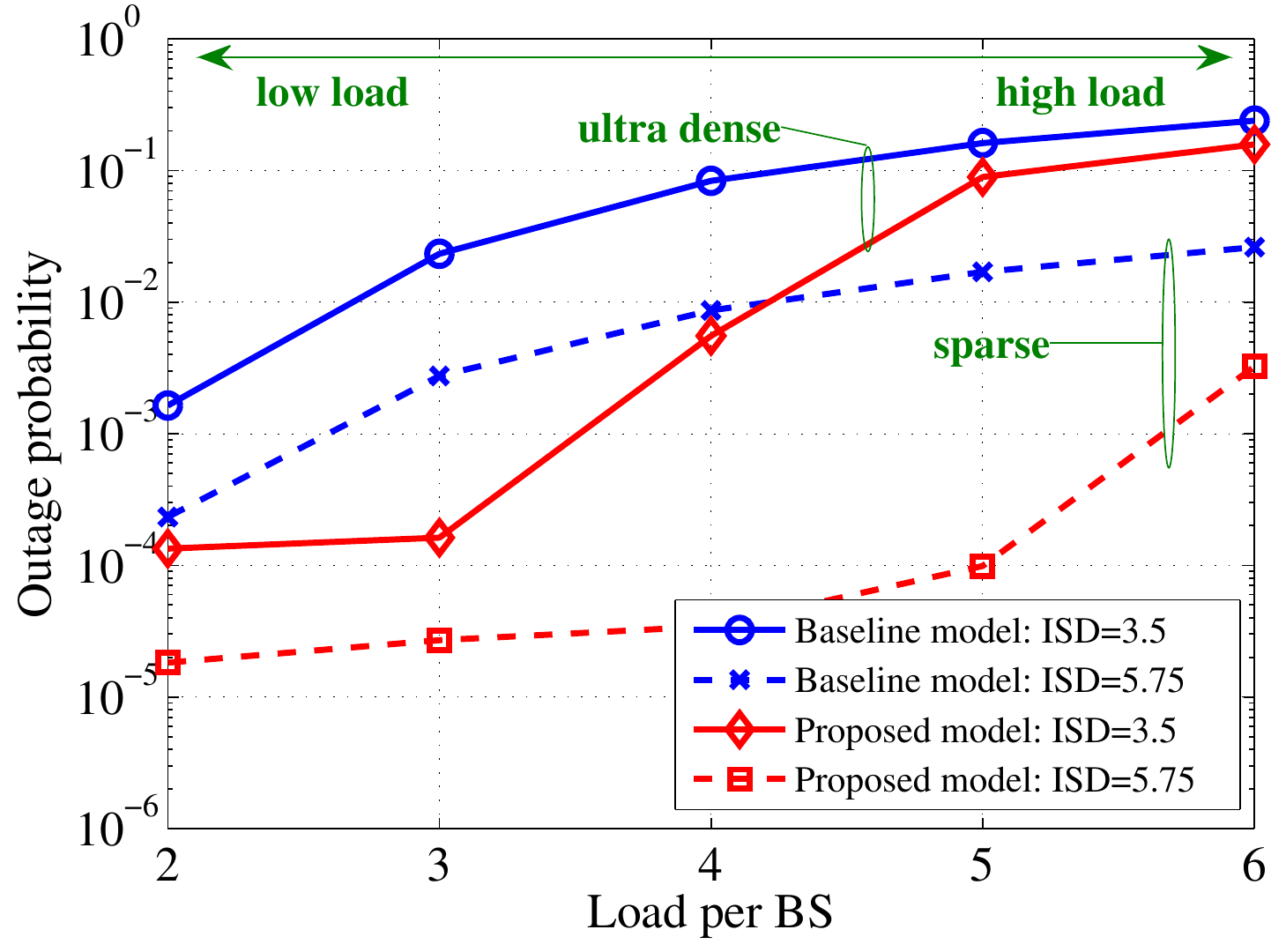}
		\label{fig:outage_ueChange}
	}
	\caption{Comparison of the behavior of EE and outage probabilities for different loads.}
	\label{fig:ue_change}
\end{figure}

Fig.~\ref{fig:bpe_ueChange} shows how EE varies with the number of UEs per SBS (i.e., load).
We can observe that the proposed method shows higher EE gains for dense networks.
For $k=6$ UEs per SBS and $\isd=3.5$, the proposed method yields an EE gain of $70.7\%$ compared to the baseline model while reaching $20.3\%$ as the network becomes sparse with $\isd=5.75$.
The gains in the proposed method for dense scenarios are due to the adaptive nature of transmit power and the QSI aware optimal UE scheduler.
As the load per SBS decreases, UEs are scheduled more often and it increases the average rate per UE.
Therefore, an improvement in EE is observed in both methods as illustrated in Fig.~\ref{fig:bpe_ueChange}.
Furthermore, as the rates increase, the myopic approach of the baseline method becomes efficient and thus, for the dense network $\isd=3.5$ with $k=2$, EE gain of the proposed method over the baseline method decreases to $4.8\%$.
Moreover, for the sparse scenario with $\isd=5.75$ and $k=2$, the EE of the baseline method exceeds the proposed method in which an EE loss of $12.6\%$ is observed in the proposed method.

In Fig.~\ref{fig:outage_ueChange}, we compare the outage probabilities for different loads.
In this figure, we can see that the outages increase for both proposed and baseline methods as the load increases.
Due to the increased number of UEs, each UE has to wait longer before it is scheduled.
Thus, there is a higher chance that the arrivals are dropped as the SBS queue capacity is exceeded thus yielding further outages.
These outages are low for sparse networks but they become significantly large for UDNs due to the increased interference and low rates.
Moreover, for a low load scenario $k=2$, Fig.~\ref{fig:outage_ueChange} shows that the proposed method reduces the outages by $91.8\%$ and $92.2\%$ compared to the baseline model for ultra dense and sparse networks, respectively.
As the load increases to $k=6$, although both models experience high outages, the proposed model results in  $33.7\%$ and $87.6\%$ outage reductions compared to the baseline model for ultra dense ($\isd=3.5$) and sparse ($\isd=5.75$) scenarios, respectively.

Based on the above discussion, it can be noted that higher energy efficiency gains of the proposed method over the baseline model are seen as the network become dense in both SBSs and UEs.
According to the network specifications given in \cite{online:3gpp13}, a sparse network consists of about $10$ SBSs/$\text{km}^2$ (ISD $=12$) each SBS serving $5\sim 10$ UEs ($K=\{5,10\}$) while a dense network has about $95$ SBSs/$\text{km}^2$ (ISD $=4$)  and $K=\{5,10\}$.
This paper analyzes networks consisting of $60$ SBSs/$\text{km}^2$ (ISD $=6.5$) with $K=1$ UEs per SBS (close to sparse as per \cite{online:3gpp13}), $60$ SBSs/$\text{km}^2$ (ISD $=6.5$) with $K\geq 3$ UEs per SSBS (neither sparse nor dense networks), and $250$ SBSs/$\text{km}^2$ (ISD $=3.5$) with $K=6$ UEs per SBS (ultra-dense compared to \cite{online:3gpp13}).
The results show that the proposed method is applicable for networks with average densities to ultra-dense networks.

\subsection{Transmit power and UE rate comparisons}

\begin{figure}[!t]
	\centering
	\subfloat[CDF of SBS transmit power for low and high loads $k=\{2,5\}$, respectively.]{
		\includegraphics[width=\myfigfactorx\columnwidth]{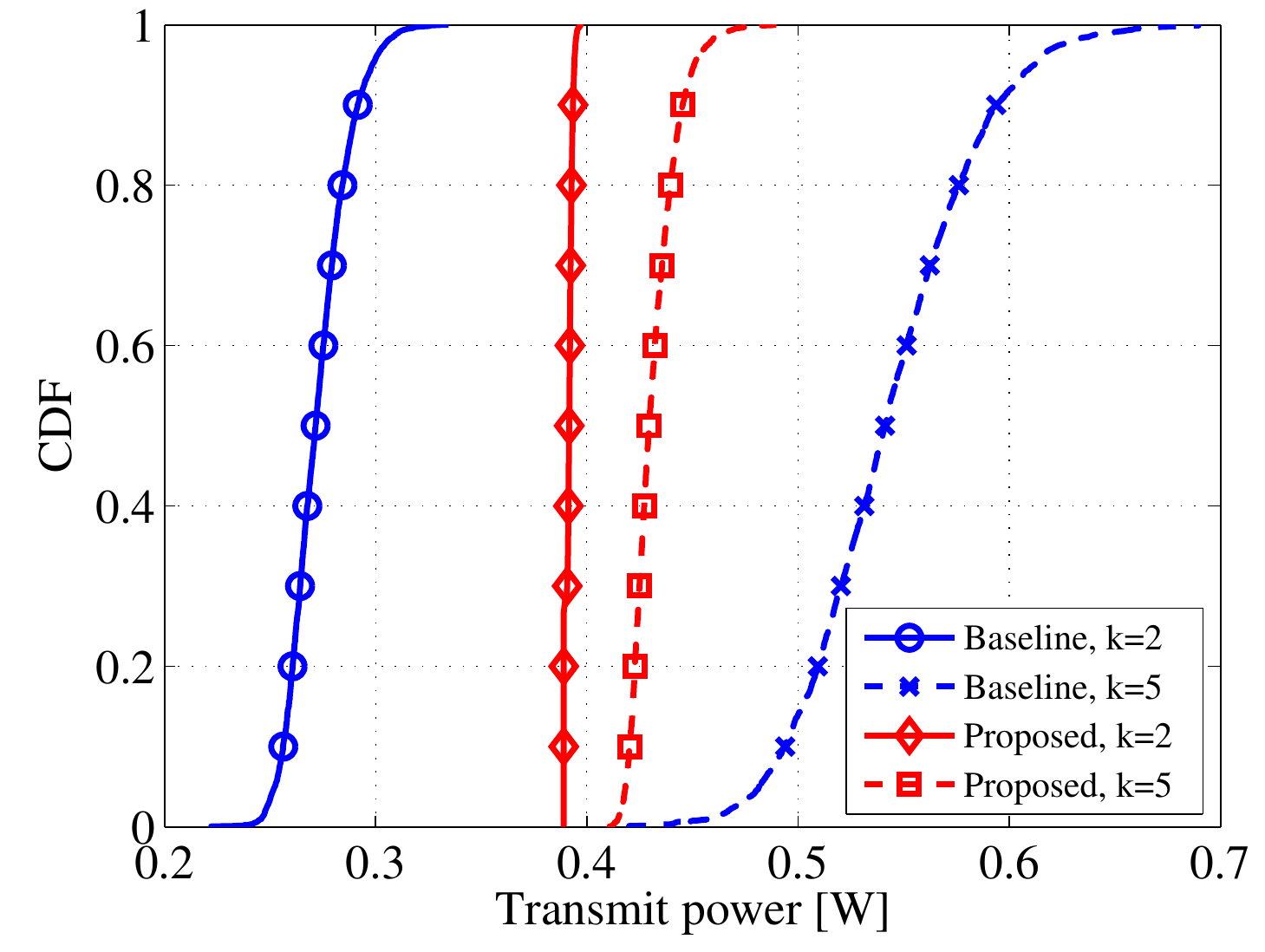}
		\label{fig:powerCDF_SPS}
	}
	\hfil
	\subfloat[CDF of UE rates for low and high loads $k=\{2,5\}$, respectively.]{
		\includegraphics[width=\myfigfactorx\columnwidth]{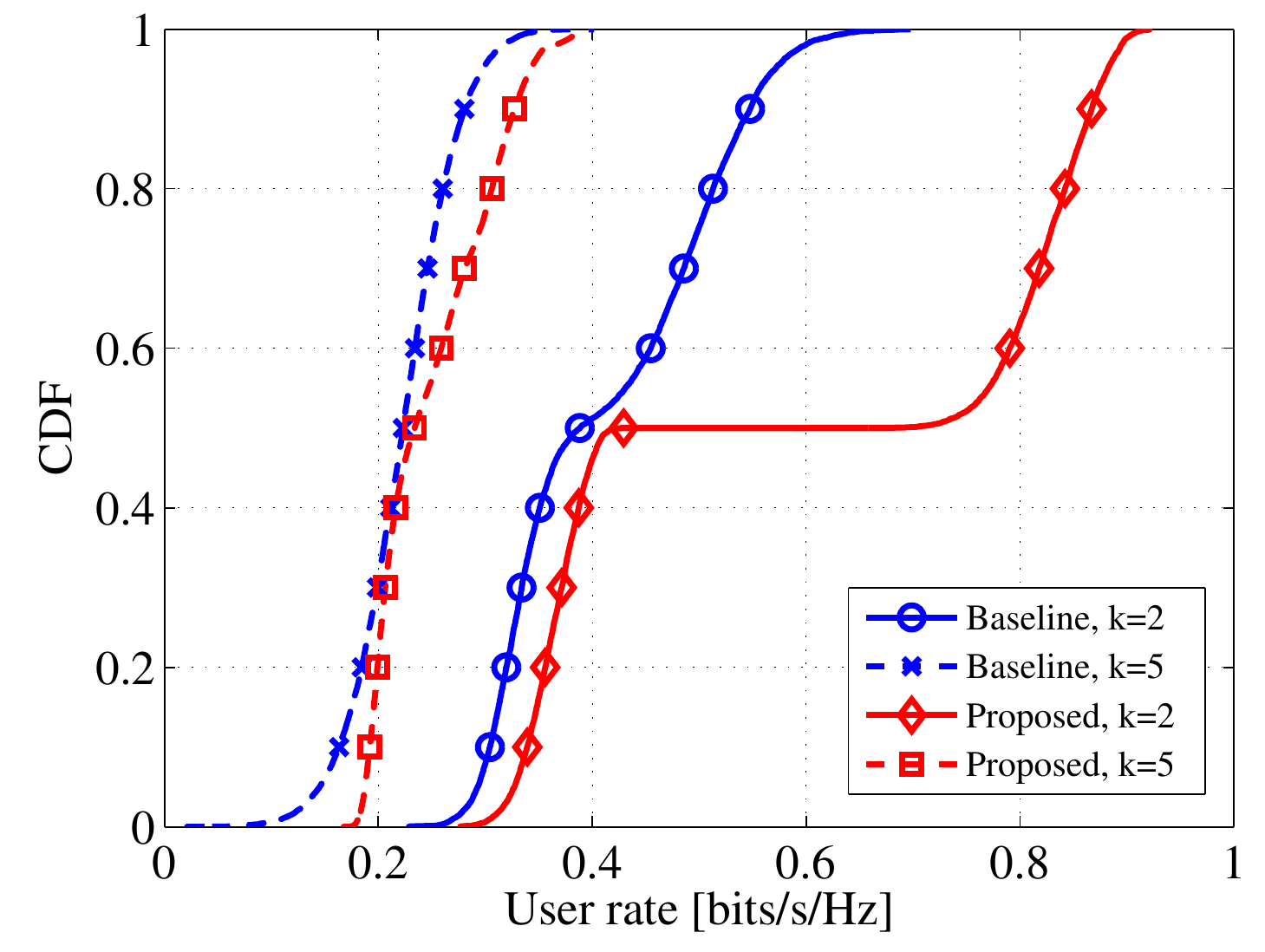}
		\label{fig:rateCDF_SPS}
	}
	\caption{Comparison of the SBS transmit powers and UE rates for a sparse scenario ($\isd=5.75$).}
	\label{fig:CDF_SPS}
\end{figure}

Fig. \ref{fig:CDF_SPS} presents the cumulative density functions (CDFs) of SBS transmit powers and UE rates for a sparse scenario with $\isd=5.75$.
Fig. \ref{fig:powerCDF_SPS} shows that the baseline method uses low transmit power compared to the proposed method for a low load with $k=2$.
The low load and the low interference in a sparse scenario allows all SBSs in the baseline method to make accurate estimations of the interferences.
Thus, SBSs in the baseline method can efficiently solve their instantaneous EE maximization problem thus resulting in higher EE compared to the proposed method.
This advantage of the baseline method reaches up to $43\%$.
As the load increases, the baseline method will begin to consume a higher amount of transmit power to satisfy the QoS of UEs while the proposed method manges to achieve its EE maximization goal with a small increase in the transmit power.
Thus, for highly loaded networks, the proposed method exhibits a $20.5\%$ reduction of the transmit power over the baseline model.
In Fig.~\ref{fig:rateCDF_SPS}, we can also see that the proposed method yields an improvement in the overall UE rates, when compared to the baseline. 
This performance advantage, in terms of rate, reaches up to $43.6\%$ and $12.3\%$ for $k=2$ and $k=5$ UEs per SBS, respectively.
Clearly, when jointly considering EE and rates, the proposed approach provides a significant overall improvement of UDN performance, when compared to the baseline.

\begin{figure}[!t]
	\centering
	\subfloat[CDF of SBS transmit power for low and high loads $k=\{2,5\}$, respectively.]{
		\includegraphics[width=\myfigfactorx\columnwidth]{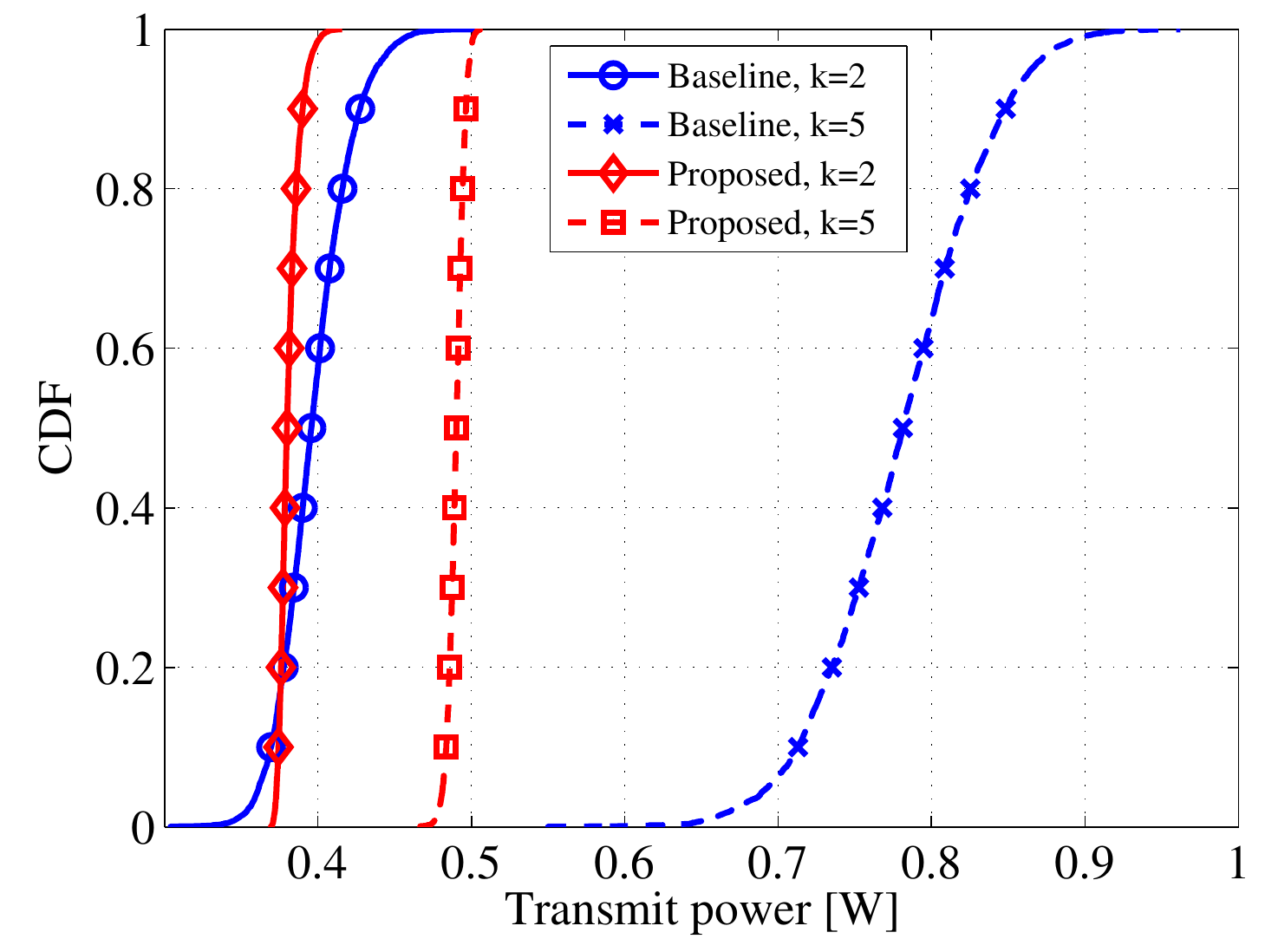}
		\label{fig:powerCDF_UDN}
	}
	\hfil
	\subfloat[CDF of UE rates for low and high loads $k=\{2,5\}$, respectively.]{
		\includegraphics[width=\myfigfactorx\columnwidth]{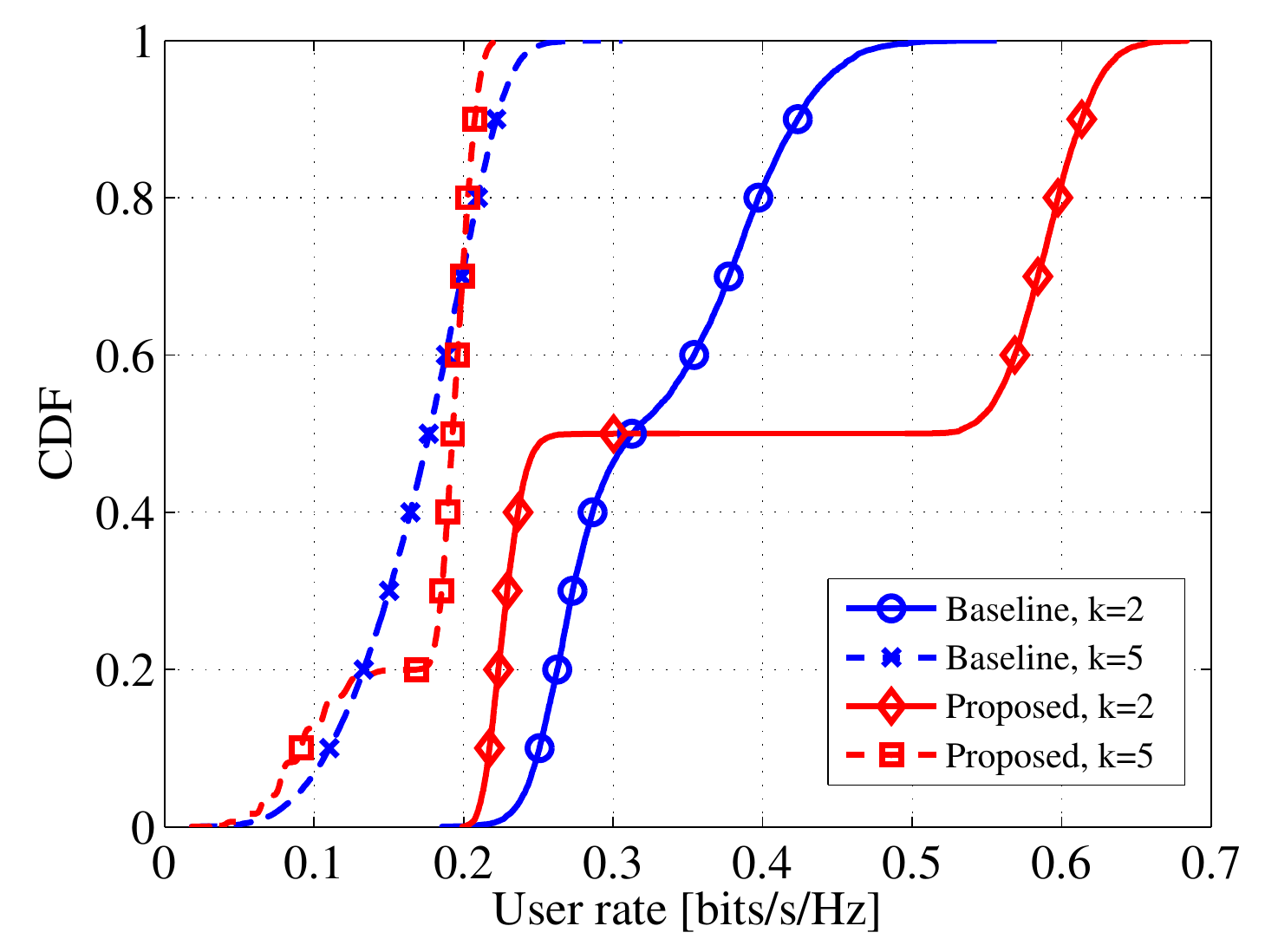}
		\label{fig:rateCDF_UDN}
	}
	\caption{Comparison of the SBS transmit powers and UE rates for a dense scenario ($\isd=3.5$).}
	\label{fig:CDF_UDN}
\end{figure}

Fig. \ref{fig:CDF_UDN} presents the CDFs of SBS transmit powers and UE rates for a dense scenario with $\isd=3.5$.
From Fig. \ref{fig:powerCDF_UDN}, it can be observed that the proposed method uses low transmit power and provides higher UE rates on the average compared to the baseline model.
For the low load $k=2$, the average transmit powers are 0.381 W and 0.397 W for the proposed and baseline methods, respectively.
To handle the excess load, the respective average transmit powers increase to 0.49 W and 0.78 W as the load is increased to $k=5$.
Therefore, we can see that the proposed method reduces the power consumption by about $4\%$ for low loads and $37.1\%$ for high loads compared to the baseline model, respectively.
The average UE rates of proposed and baseline methods are 0.41 bits/s/Hz and 0.33 bits/s/Hz for low loads while 0.18~bits/s/Hz and 0.17 bits/s/Hz for high loads, respectively as shown in Fig.~\ref{fig:rateCDF_UDN}.
Thus, the gains in proposed method over the baseline method are $24.4\%$ and $2.4\%$ for low and high loads, respectively.
These gains validates the suitability of the proposed method on the dense networks.

Based on the above comparisons, it can be observed that the proposed method which adapts to the dynamics in the network is able to reduce energy consumption while maintaining a desirable balance between EE and data rate. 
Moreover, the proposed approach is able to reduce the number of UEs that are in outage due to the limitation on the SBS queue capacity.
For UDNs, the performance improvements become much more significant, reaching up to a $70.7\%$ improvement in EE and up to $33.7\%$ reduction of outages.
These results demonstrate that the proposed solution can clearly turn the density of UDNs into gains in terms of energy efficiency.

\subsection{Impact of the boundary conditions}\label{subsec:impact_boundry}

To derive a solution for the MF framework one must deal with solving HBJ-FPK PDEs in which the final solution will depend on predefined boundary conditions~\cite{book:gueant11}.
In this work, the choice of boundary condition is $\Gamma\big(T,\vectxx\big)=-4\exp\big(\vectxx\big)$.
At the end of the scheduling period, $t=T$, depending on the QSI, a cost is introduced based on the above boundary condition.
The cost is minimum if the queue is empty $\big(\vectx(T)=0\big)$, and grows exponentially as QSI at $T$ increases.
Thus, the choice of this boundary condition forces the SBSs to obtain smaller QSI at $t=T$.
To illustrate the impact of the boundary condition, we compare a selected set of results for three different boundary conditions, namely, 
\emph{i)} \textbf{exponential bound}: the boundary condition used in the former discussion, i.e.  $\Gamma\big(T,\vectxx\big)=-4\exp\big(\vectxx\big)$ for $\vectxx\in[0,1]$,
\emph{ii)} \textbf{uniform bound}: a boundary condition where the utility is uniform disregarding the final QSI, i.e.  $\Gamma\big(T,\vectxx\big)=-4$ for all $\vectxx\in[0,1]$, and
\emph{iii)}	\textbf{linear bound}: a boundary condition where the utility forces QSI at $t=T$ to be zero in a linear manner, i.e. $\Gamma\big(T,\vectxx\big)=-4(e-1)\vectxx -4$ for $\vectxx\in[0,1]$.
These boundary conditions are illustrated in Fig.~\ref{fig:boundary_conditions}.
\begin{figure}[!t]
	\centering
	\includegraphics[width=\myfigfactorx\columnwidth]{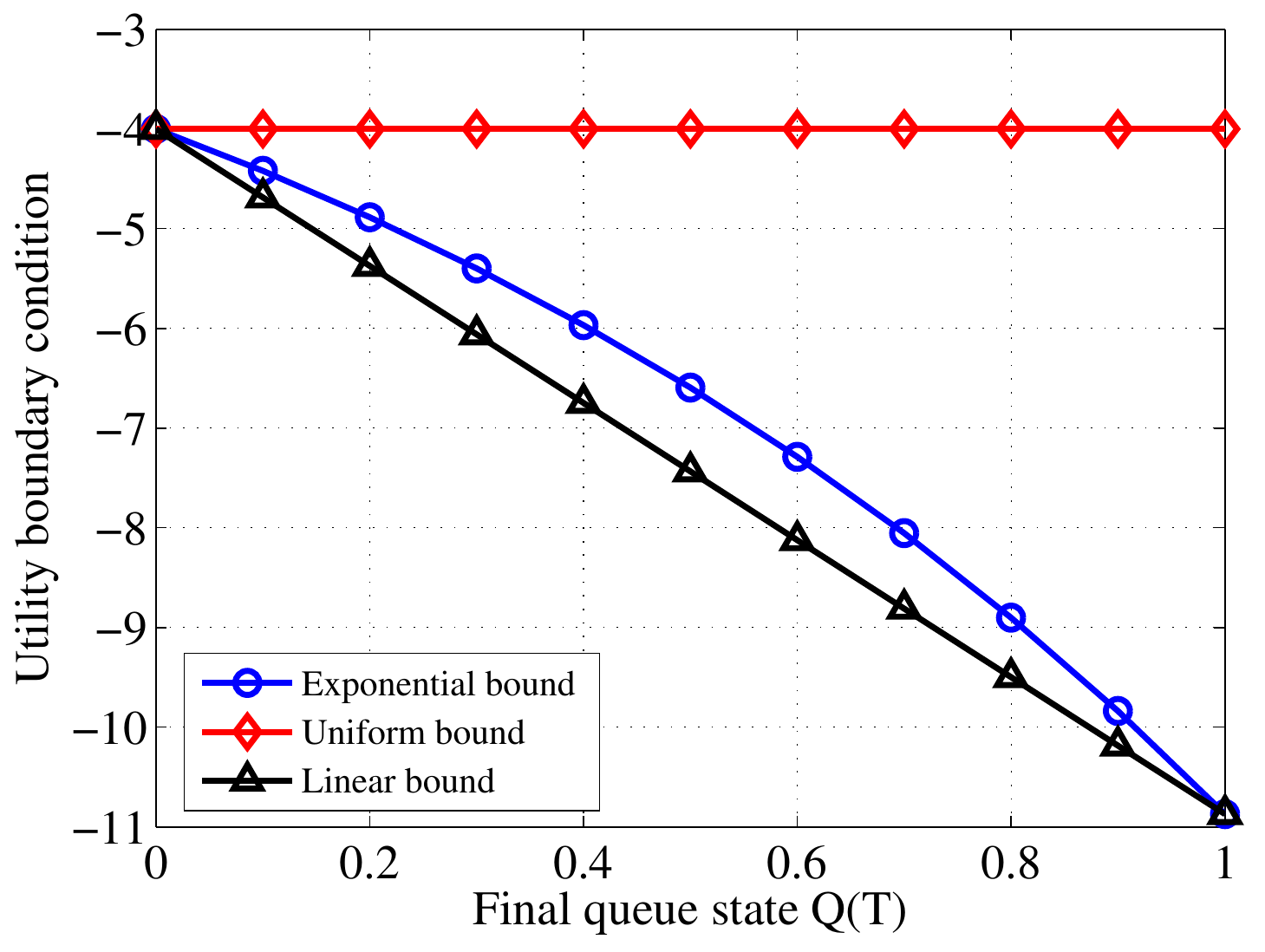}
	\caption{Different boundary conditions used for the utility at the end of scheduling period.}
	\label{fig:boundary_conditions}
\end{figure}
Here, the uniform bound is a relaxed boundary condition compared to the original exponential boundary condition as it grants fixed utility disregarding the QSI at $t=T$.
On the contrary, the linear bound forces QSI at $t=T$ to be even closer to zero compared to the exponential bound and thus, can be considered as a tighter boundary condition.

The MF distributions for the above boundary conditions are given in Fig.~\ref{fig:MF_bounds_man}.
\begin{figure}[!t]
	\centering
	\subfloat[With the exponential boundary.]{
		\includegraphics[width=\myfigfactorx\columnwidth]{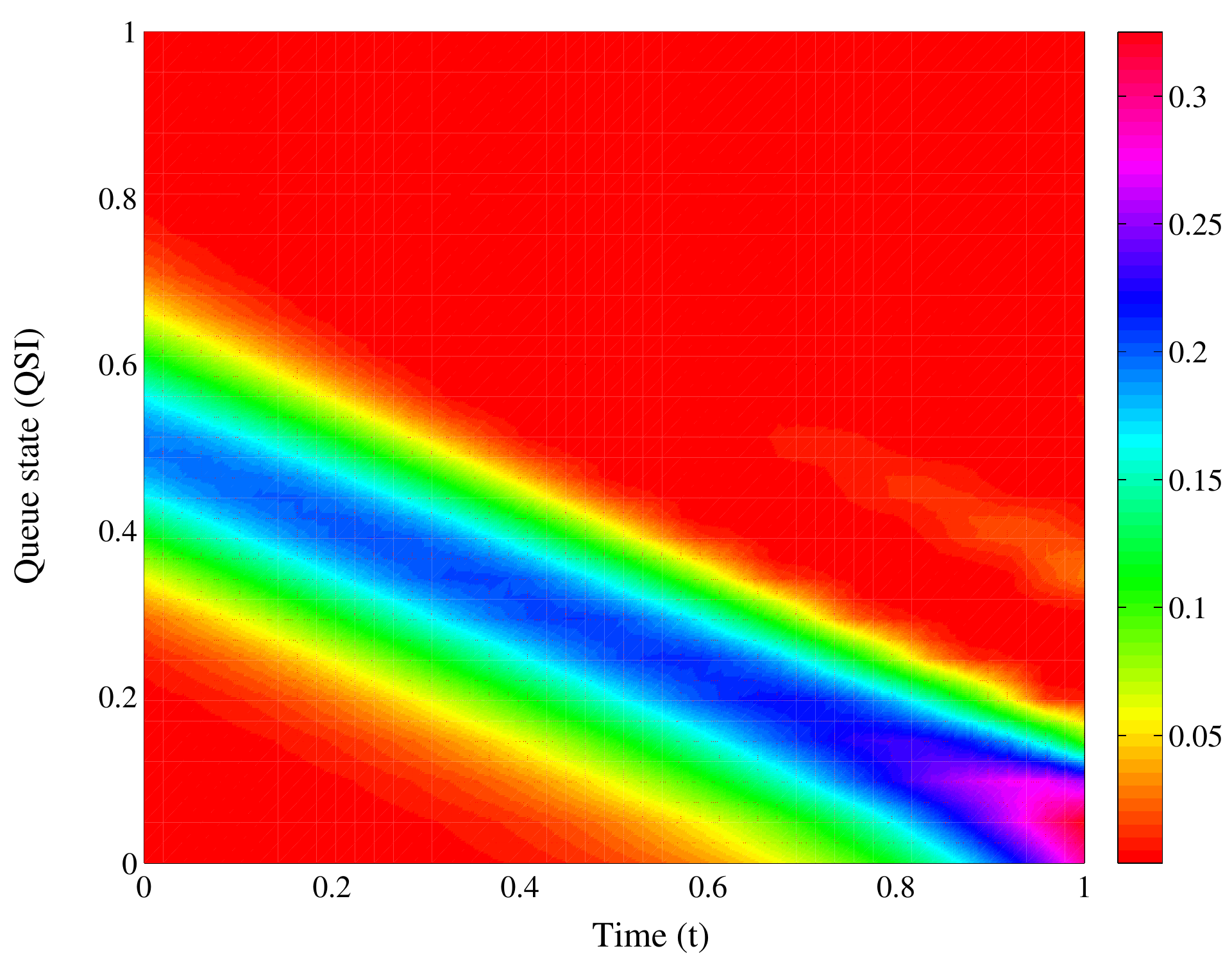}
		\label{fig:MF_exp_man}
	}
	\hfil
	\subfloat[With the uniform boundary.]{
		\includegraphics[width=\myfigfactorx\columnwidth]{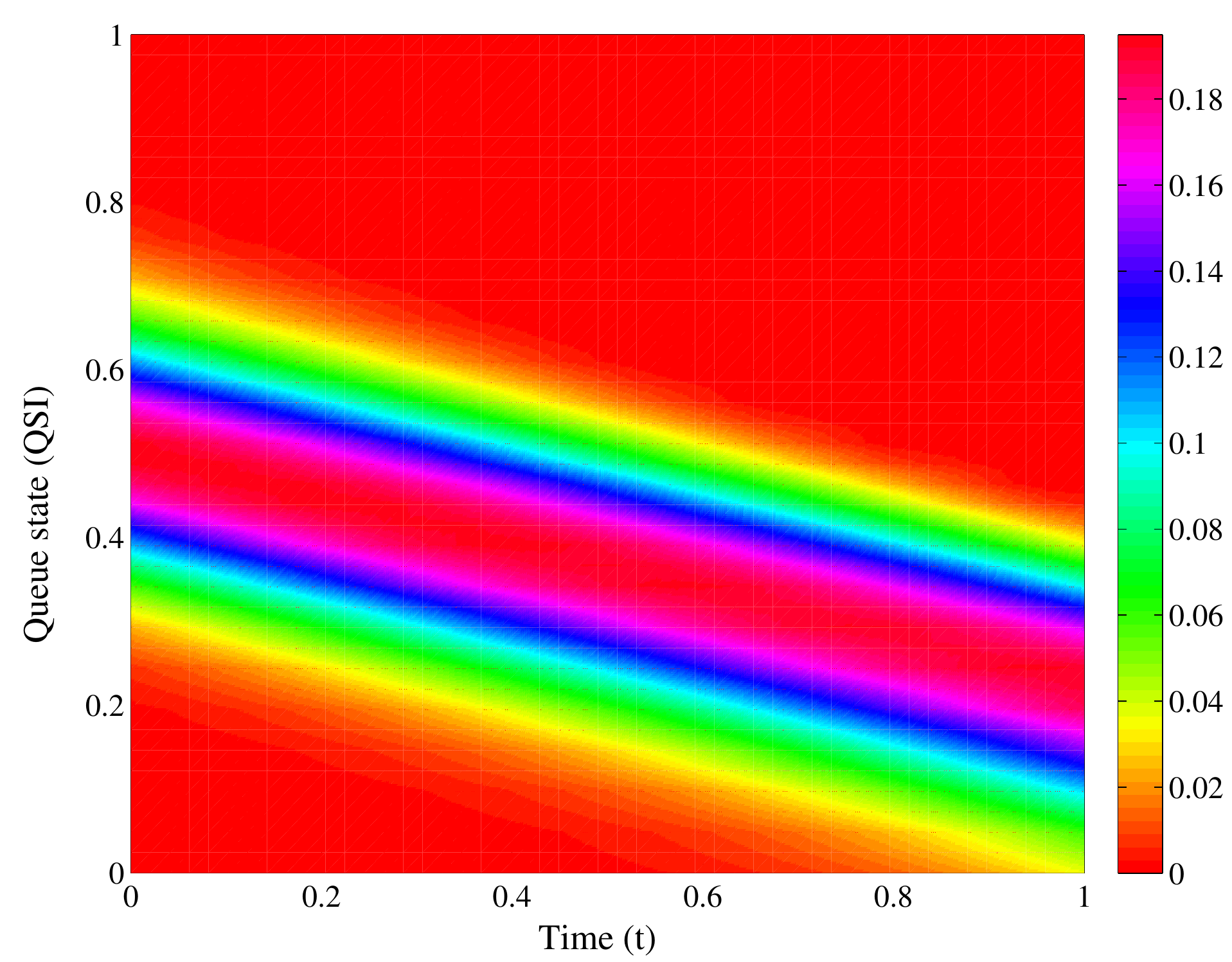}
		\label{fig:MF_uni_man}
	}
	\hfil
	\subfloat[With the linear boundary.]{
		\includegraphics[width=\myfigfactorx\columnwidth]{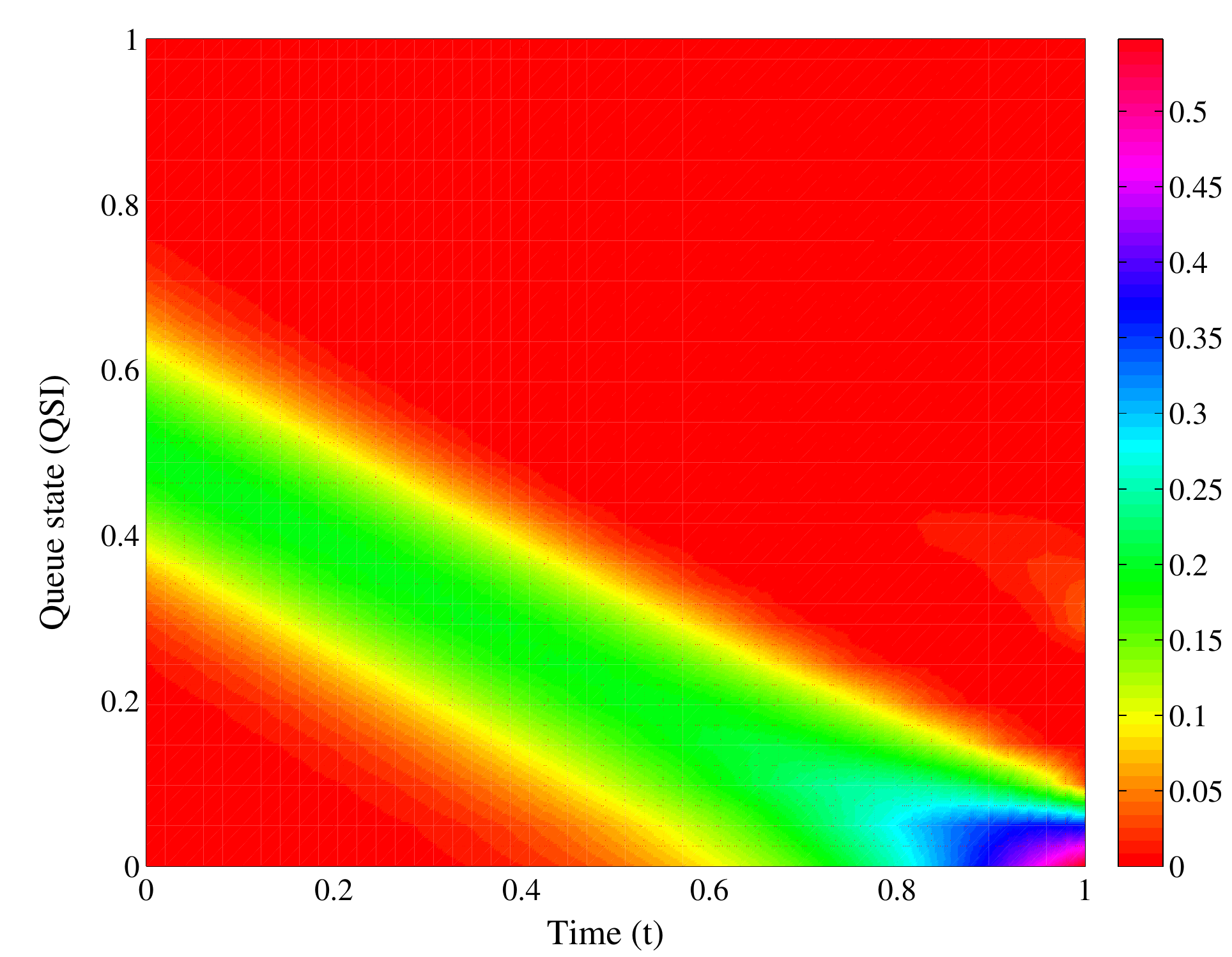}
		\label{fig:MF_lin_man}
	}
	\caption{MF distributions for different boundary conditions.}
	\label{fig:MF_bounds_man}
\end{figure}
It can be noted that with the uniform bound, that a larger fraction of users with non-zero queues can be seen compared to the exponential boundary. 
The main reason is that the fixed utility introduced by the boundary condition does not force QSI to be zero by the end of the scheduling period.
On the contrary, the linear bound encourages to obtain almost empty queues by $t=T$ and thus, a large fraction of UEs with zero queues can be observed.

Finally, the comparison of EE and outage probabilities for the above boundary conditions are illustrated in Fig. \ref{fig:MF_bounds_perf_man}.
\begin{figure}[!t]
	\centering
	\subfloat[Average EE for different ISDs.]{
		\includegraphics[width=\myfigfactorx\columnwidth]{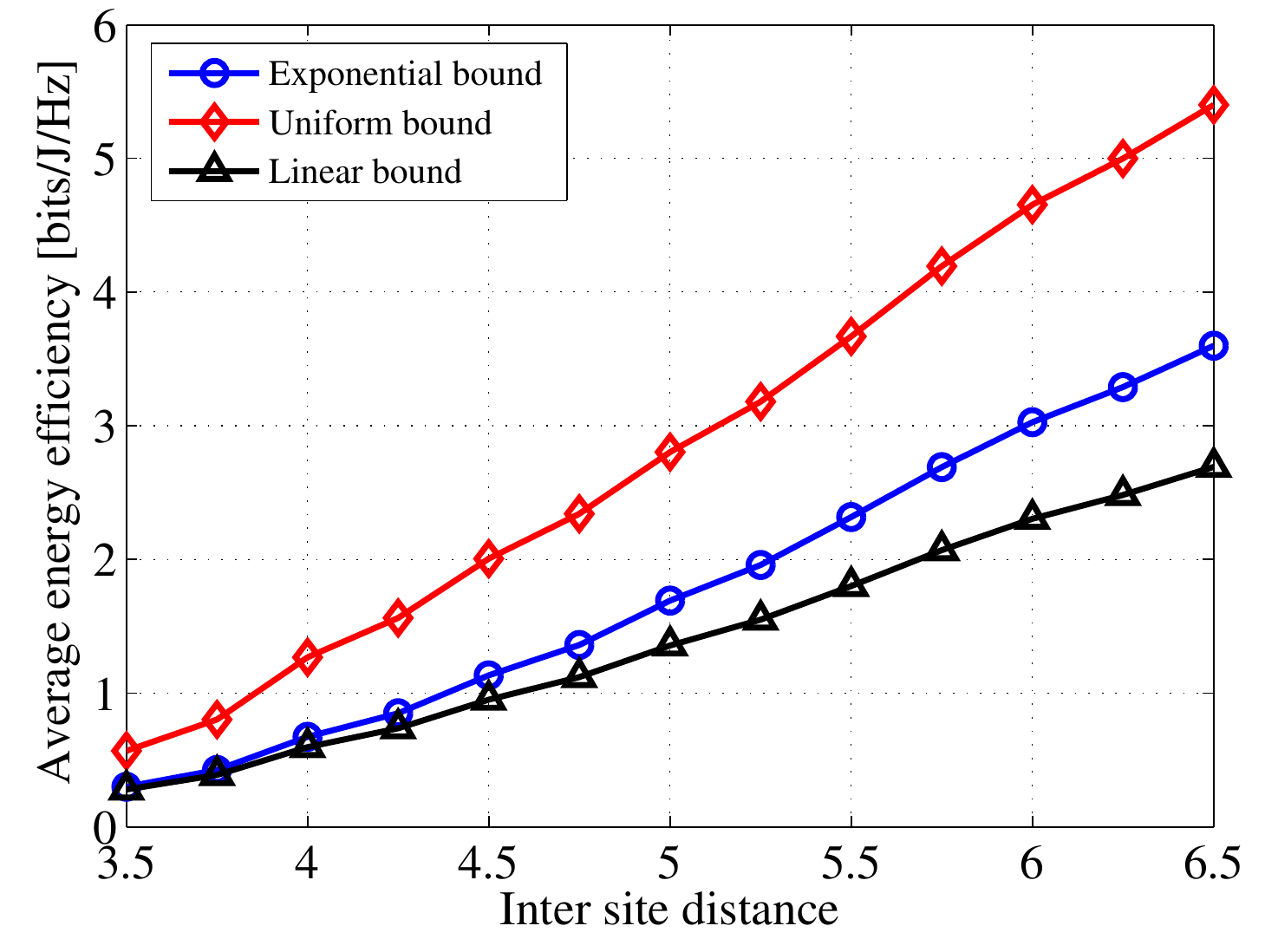}
		\label{fig:EE_bounds_man}
	}
	\hfil
	\subfloat[Outage probabilities for different ISDs]{
		\includegraphics[width=\myfigfactorx\columnwidth]{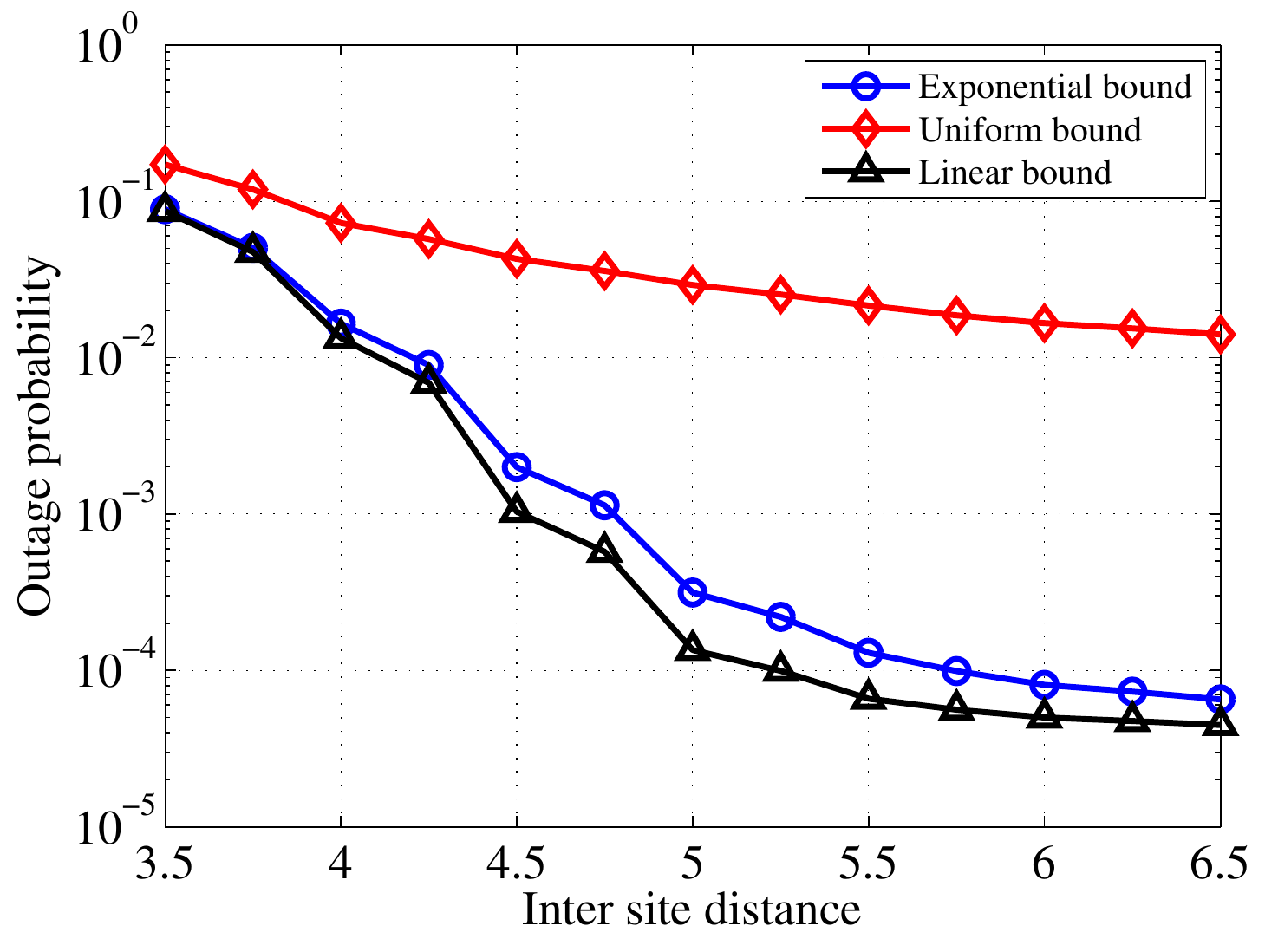}
		\label{fig:OP_bounds_man}
	}
	\caption{EE and outage probability comparison for different boundary conditions. The load is $k=5$. }
	\label{fig:MF_bounds_perf_man}
\end{figure}
Here, compared to the exponential bound, the uniform bound provides higher EE at the price of higher outage probabilities.
Relaxing the QSI at $t=T$ allows SBSs to use low transmit power and increase the EE.
However, the accumulated queues prevent further arrivals and thus, higher outage probabilities can be observed.
With the linear bound, all the data in the queues are transferred to UEs by the end of the scheduling period.
Thus, more arrivals can be queued resulting reduced outage probabilities over both other boundary conditions.
However, to ensure the empty queues by $t=T$, SBSs use higher energy to increase the data rates and thus, the EE is reduced compared to the exponential bound.

\section{Conclusions}\label{sec:conclusion}
In this paper, we have formulated the problem of joint power control and user scheduling for ultra-dense small cell deployment as a MFG under the uncertainties of QSI and CSI.
The goal is to maximize a time-average utility (energy efficiency in terms of bits per unit power) while ensuring users' QoS concerning outages due to queue capacity.
Under appropriate assumptions, we have analyzed the equilibrium of the MFG with the aid of low-complex tractable two partial differential equations (PDEs).
In conjunction the with mean-field approach, we have introduced a Lyapunov-based QSI and CSI aware user scheduler.
Using simulations, we have shown that the proposed method provides considerable gains in EE and massive reductions in outages compared to a baseline model and thus, becomes a suitable candidate for future UDNs.
One key future extension for this work is to account explicitly for SBS sleep mode optimization and adaptive cell association.

\bibliographystyle{IEEEtran}
\bibliography{my_journal_MFG}

\begin{IEEEbiography}[{\includegraphics[width=1in,height=1.25in,clip,keepaspectratio]{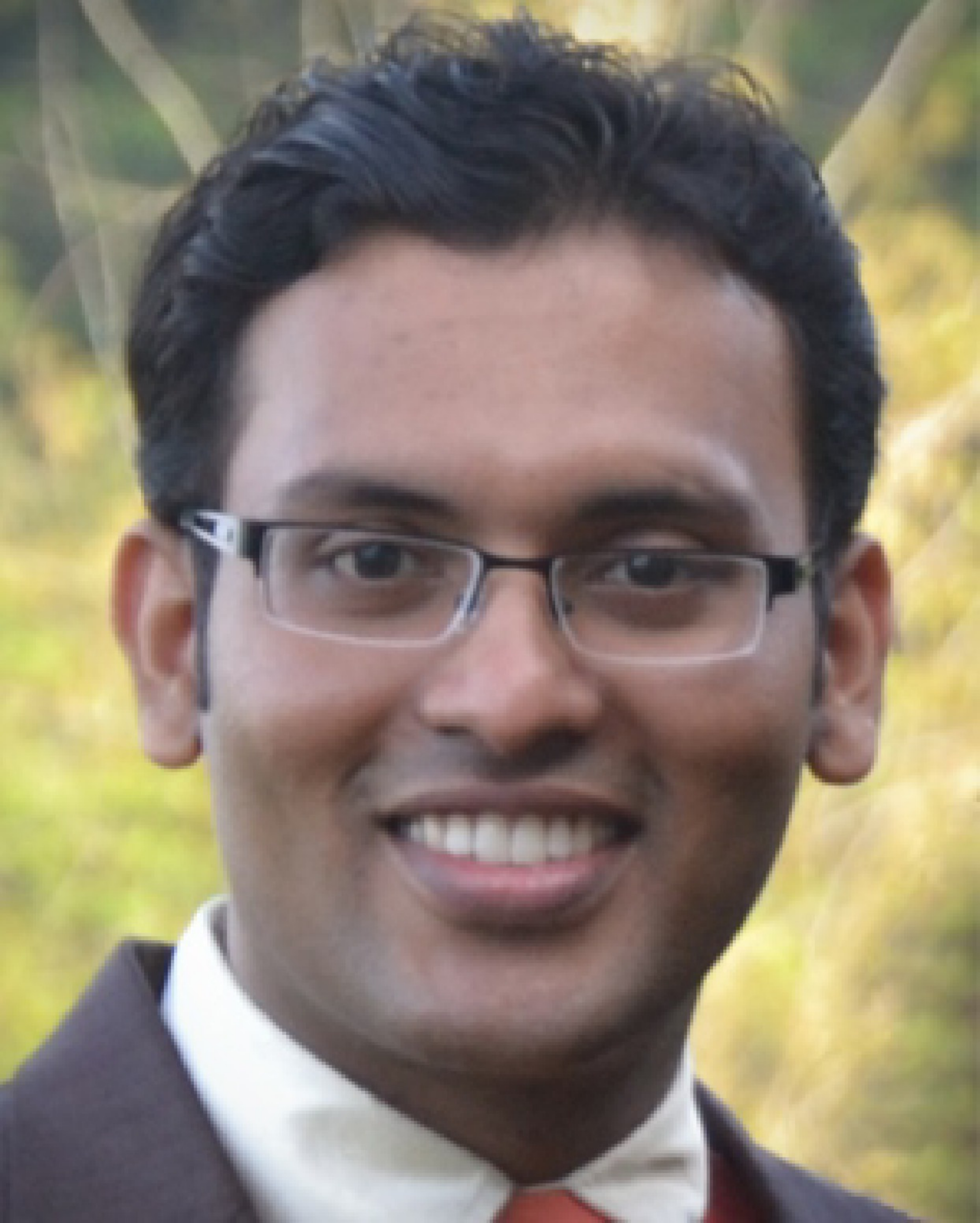}}]{Sumudu Samarakoon}
	received his B. Sc. (Hons.) degree in Electronic and Telecommunication Engineering from the University of Moratuwa, Sri Lanka and the M. Eng. degree from the Asian Institute of Technology, Thailand in 2009 and 2011, respectively.
	He is currently pursuing hid Dr. Tech degree in Communications Engineering at the University of Oulu, Finland.
	Sumudu is also a member of the research staff of the Centre for Wireless Communications (CWC), Oulu, Finland.
	His main research interests are in heterogeneous networks, radio resource management, machine learning, and game theory.
\end{IEEEbiography}
\vfill

\begin{IEEEbiography}[{\includegraphics[width=1in,height=1.25in,clip,keepaspectratio]{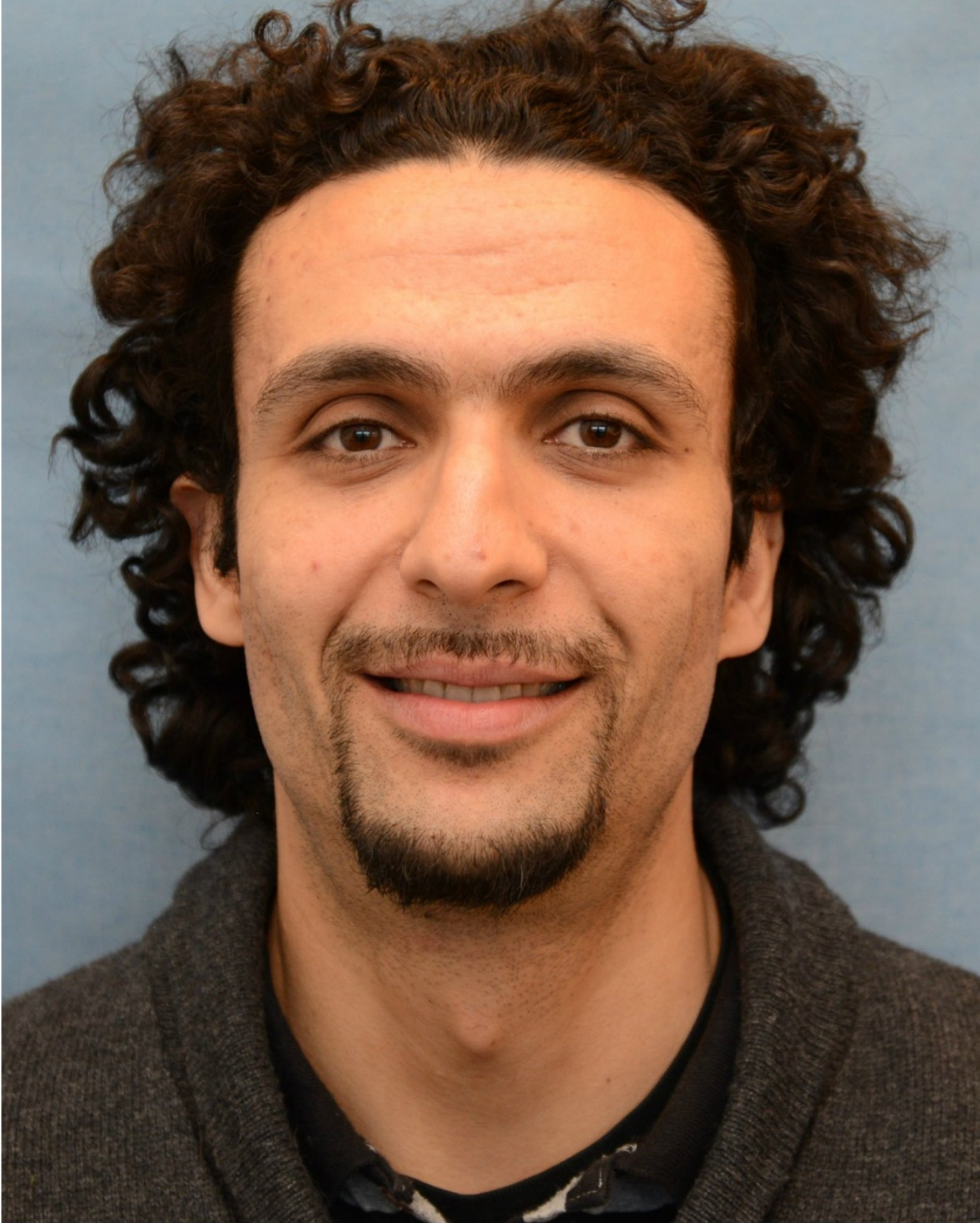}}]{Mehdi Bennis}
	received his M.Sc. degree in Electrical Engineering jointly from the EPFL, Switzerland and the Eurecom Institute, France in 2002.
	From 2002 to 2004, he worked as a research engineer at IMRA-EUROPE investigating adaptive equalization algorithms for mobile digital TV.
	In 2004, he joined the Centre for Wireless Communications (CWC) at the University of Oulu, Finland as a research scientist.
	In 2008, he was a visiting researcher at the Alcatel-Lucent chair on flexible radio, SUPELEC.
	He obtained his Ph.D. in December 2009 on spectrum sharing for future mobile cellular systems.
	
	His main research interests are in radio resource management, heterogeneous networks, game theory and machine learning in 5G networks and beyond.
	He has co-authored one book and published more than 100 research papers in international conferences, journals and book chapters.
	He was the recipient of the prestigious 2015 Fred W. Ellersick Prize from the IEEE Communications Society.
	Dr. Bennis serves as an editor for the IEEE Transactions on Wireless Communications.
\end{IEEEbiography}
\vspace{-5mm}
\begin{IEEEbiography}[{\includegraphics[width=1in,height=1.25in,clip,keepaspectratio]{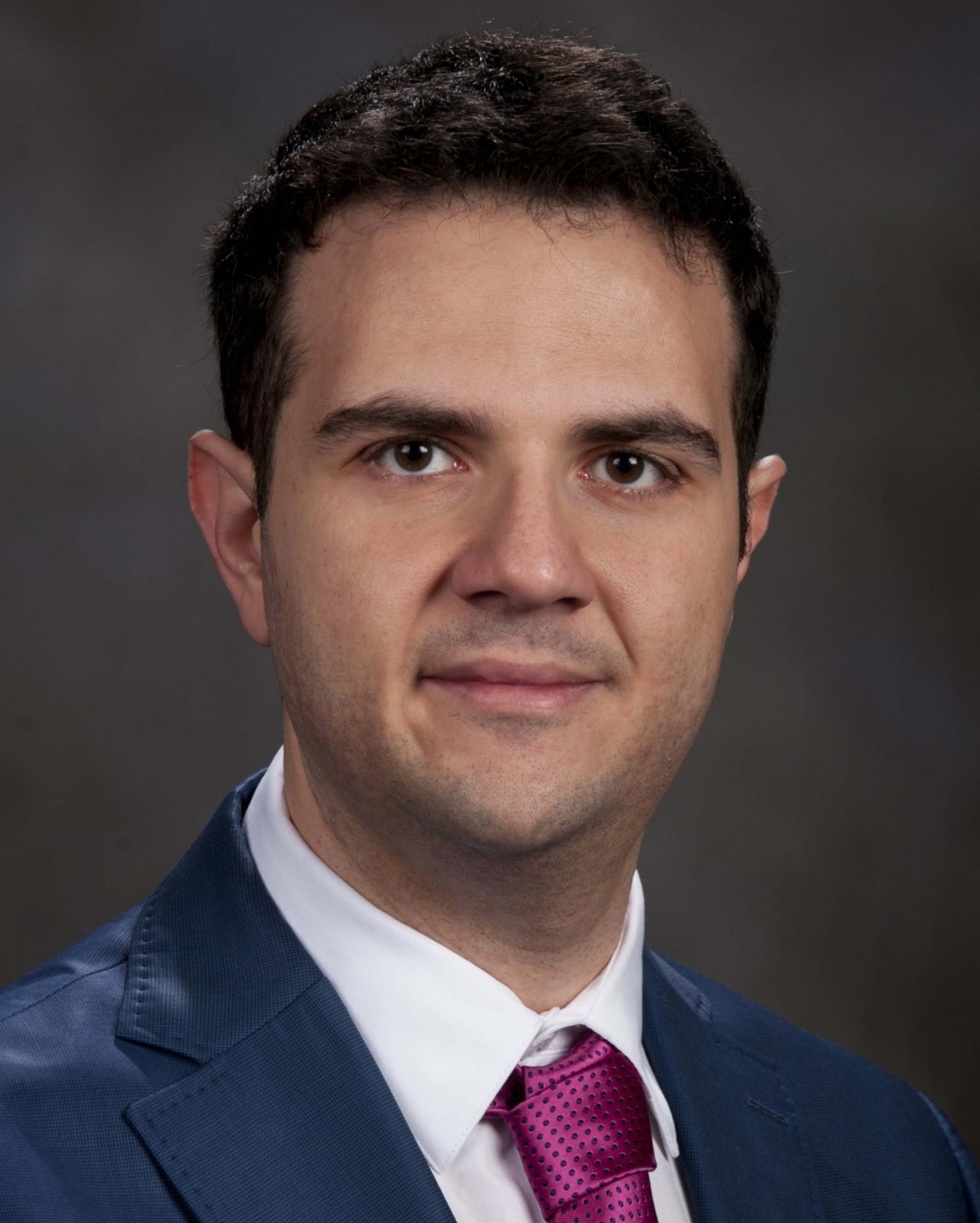}}]{Walid Saad}
	(S'07, M'10, SM’15) Walid Saad received his Ph.D degree from the University of Oslo in 2010. 
	Currently,  he is an Assistant Professor and the Steven O. Lane Junior Faculty Fellow at the Department of Electrical and Computer Engineering at Virginia Tech, where he leads the Network Science, Wireless, and Security (NetSciWiS) laboratory, within the Wireless@VT research group. 
	His  research interests include wireless networks, game theory, cybersecurity, and cyber-physical systems. 
	Dr. Saad is the recipient of the NSF CAREER award in 2013, the AFOSR summer faculty fellowship in 2014, and the Young Investigator Award from the Office of Naval Research (ONR) in 2015. 
	He was the author/co-author of five conference best paper awards at WiOpt in 2009, ICIMP in 2010, IEEE WCNC in 2012,  IEEE PIMRC in 2015, and IEEE SmartGridComm in 2015. 
	He is the recipient of the 2015 Fred W. Ellersick Prize from the IEEE Communications Society. 
	Dr. Saad serves as an editor for the IEEE Transactions on Wireless Communications, IEEE Transactions on Communications, and IEEE Transactions on Information Forensics and Security.
\end{IEEEbiography}
\vspace{-5mm}
\begin{IEEEbiography}[{\includegraphics[width=1in,height=1.25in,clip,keepaspectratio]{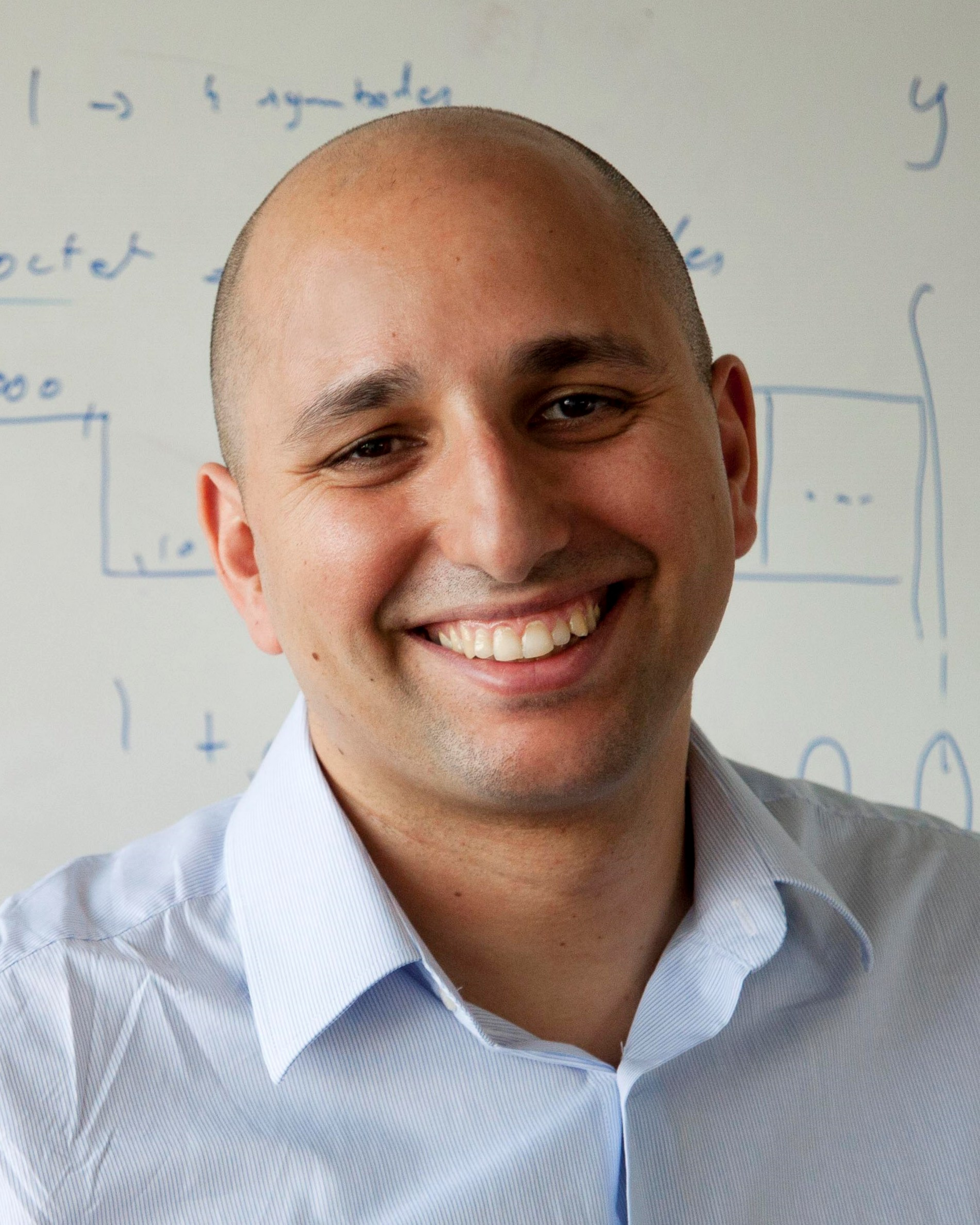}}]{M\'{e}rouane Debbah}
	(S'01-AM'03-M'04-SM'08- 1030 F'15) entered the Ecole Normale Sup\'{e}rieure de Cachan (France) in 1996 where he received his M.Sc and Ph.D. degrees respectively. 
	He worked for Motorola Labs (Saclay, France) from 1999-2002 and the Vienna Research Center for Telecommunications (Vienna, Austria) until 2003. 
	From 2003 to 2007, he joined the Mobile Communications department of the Institut Eurecom (Sophia Antipolis, France) as an Assistant Professor. 
	Since 2007, he is a Full Professor at CentraleSupelec (Gif-sur-Yvette, France). 
	From 2007 to 2014, he was the director of the Alcatel-Lucent Chair on Flexible Radio. Since 2014, he is Vice-President of the Huawei France R\&D center and director of the Mathematical and Algorithmic Sciences Lab. His research interests lie in fundamental mathematics, algorithms, statistics, information \& communication sciences research. 
	He is an Associate Editor in Chief of the journal Random Matrix: Theory and Applications and was an associate and senior area editor for IEEE Transactions on Signal Processing respectively in 2011-2013 and 2013-2014. 
	M\'{e}rouane Debbah is a recipient of the ERC grant MORE (Advanced Mathematical Tools for Complex Network Engineering). 
	He is a IEEE Fellow, a WWRF Fellow and a member of the academic senate of Paris-Saclay. 
	He has managed 8 EU projects and more than 24 national and international projects. 
	He received 14 best paper awards, among which the 2007 IEEE GLOBECOM best paper award, the Wi-Opt 2009 best paper award, the 2010 Newcom++ best paper award, the WUN CogCom Best Paper 2012 and 2013 Award, the 2014 WCNC best paper award, the 2015 ICC best paper award, the 2015 IEEE Communications Society Leonard G. Abraham Prize and 2015 IEEE Communications Society Fred W. Ellersick Prize as well as the Valuetools 2007, Valuetools 2008, CrownCom2009, Valuetools 2012 and SAM 2014 best student paper awards. 
	He is the recipient of the Mario Boella award in 2005, the IEEE Glavieux Prize Award in 2011 and the Qualcomm Innovation Prize Award in 2012.
\end{IEEEbiography}

\vfill

\begin{IEEEbiography}[{\includegraphics[width=1in,height=1.25in,clip,keepaspectratio]{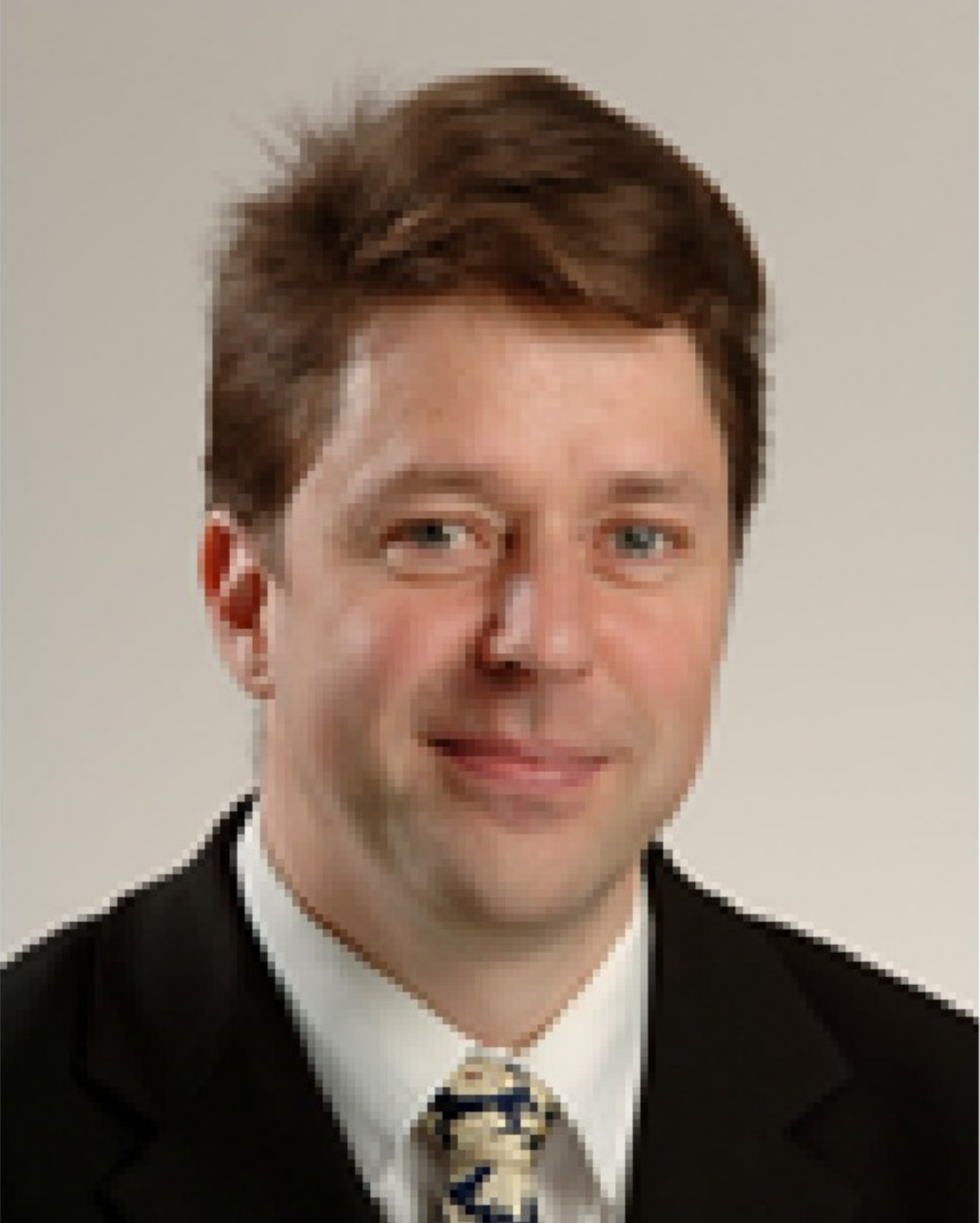}}]{Matti Latva-aho}
	received the M.Sc., Lic.Tech. and Dr. Tech (Hons.) degrees in Electrical Engineering from the University of Oulu, Finland in 1992, 1996 and 1998, respectively.
	From 1992 to 1993, he was a Research Engineer at Nokia Mobile Phones, Oulu, Finland after which he joined Centre for Wireless Communications (CWC) at the University of Oulu.
	Prof. Latva-aho was Director of CWC during the years 1998-2006 and Head of Department for Communication Engineering until August 2014.
	Currently he is Professor of Digital Transmission Techniques at the University of Oulu.
	His research interests are related to mobile broadband communication systems and currently his group focuses on 5G systems research.
	Prof. Latva-aho has published 300+ conference or journal papers in the field of wireless communications.
\end{IEEEbiography}

\end{document}